\def\final{1}

\documentclass[11pt]{article}

\usepackage{algorithm, algorithmic}
\usepackage{amsmath, amssymb, amsthm}
\usepackage{framed}
\usepackage{fullpage}
\usepackage[margin=1.1in]{geometry}
\usepackage{verbatim}
\usepackage{color}
\usepackage{kpfonts}
\definecolor{DarkGreen}{rgb}{0.1,0.5,0.1}
\definecolor{DarkRed}{rgb}{0.6,0.2,0.2}
\definecolor{DarkBlue}{rgb}{0.2,0.2,0.6}
\usepackage[pdftex]{hyperref}
\hypersetup{
    unicode=false,				
    pdftoolbar=true,				
    pdfmenubar=true,			
    pdffitwindow=false, 		
    pdftitle={},    					
    pdfauthor={}
    pdfsubject={},				
    pdfnewwindow=true,		
    pdfkeywords={},			
    colorlinks=true,				
    linkcolor=DarkBlue,		
    citecolor=DarkBlue,	
    filecolor=DarkRed,		
    urlcolor=DarkBlue,		
}

\ifnum\final=0
\newcommand{\mynote}[1]{\marginpar{\tiny\sf #1}}
\else
\newcommand{\mynote}[1]{}
\fi
\newcommand{\jnote}[1]{\mynote{Jon: {#1}}}

\newcommand{\pr}[2]{\underset{#1}{\mathbb{P}}\left[ #2 \right]}
\newcommand{\ex}[2]{\underset{#1}{\mathbb{E}}\left[ #2 \right]}

\newcommand{\err}{\theta}

\newcommand\N{\mathbb{N}}

\newcommand{\cD}{\mathcal{D}}
\newcommand{\cE}{\mathcal{E}}

\newcommand{\cO}{\mathcal{O}}
\newcommand{\cP}{\mathcal{P}}

\newcommand{\poly}{\mathrm{poly}}

\newcommand{\bits}{\{0,1\}}
\newcommand{\pmo}{\{\pm1\}}
\newcommand{\getsr}{\gets_{\mbox{\tiny R}}}
 
\newcommand{\set}[1]{\left\{#1\right\}} 
\newcommand{\from}{:}
\newcommand{\negl}{\mathrm{negl}}
\newcommand{\eps}{\varepsilon}

\newcommand{\INDSTATE}[1][1]{\STATE\hspace{#1\algorithmicindent}}

\newcommand{\cryptoalg}[1]{\mathit{#1}}
\newcommand{\cryptoadv}[1]{\mathcal{#1}}
\newcommand{\gen}{\cryptoalg{Gen}}
\newcommand{\enc}{\cryptoalg{Enc}}
\newcommand{\dec}{\cryptoalg{Dec}}
\newcommand{\trace}{\cryptoalg{Trace}}

\newcommand{\encgen}{\cryptoalg{\gen}}
\newcommand{\encenc}{\cryptoalg{\enc}}
\newcommand{\encdec}{\cryptoalg{\dec}}

\newcommand{\ifpc}{\cryptoadv{F}}
\newcommand{\FPCmat}{C}
\newcommand{\FPCcol}{c}
\newcommand{\fpcfail}{\eps}
\newcommand{\fpcfalse}{\delta}
\newcommand{\fpcgen}{\cryptoalg{\gen}}
\newcommand{\fpctrace}{\cryptoalg{\trace}}
\newcommand{\rob}{\beta}
\newcommand{\irob}{\left( \frac12-\beta \right)}
\newcommand{\mass}{\zeta}
\newcommand{\accuracy}{\alpha}
\newcommand{\errwt}{\gamma}
\newcommand{\security}{\lambda}

\newcommand{\pirate}{\cP}
\newcommand{\oracle}{\cO}
\newcommand{\cryptogame}[1]{\mathsf{#1}}
\newcommand{\realgame}{\cryptogame{Attack}}
\newcommand{\idealgame}{\cryptogame{IdealAttack}}
\newcommand{\accuracygame}{\cryptogame{Acc}}

\newcommand{\bnpgame}{\cryptogame{NonPrivacy}}
\newcommand{\bnpattack}{\cryptogame{PrivacyAttack}}
\newcommand{\idealbnpattack}{\cryptogame{IdealPrivacyAttack}}
\newcommand{\ifpcgame}{\cryptogame{IFPC}}
\newcommand{\encadv}{\cryptoadv{B}}
\newcommand{\fpcadv}{\cryptoadv{P}}
\newcommand{\accadv}{\cryptoadv{A}}
\newcommand{\bnpadv}{\cryptoadv{A}_{\mathrm{priv}}}
\newcommand{\consistent}{\mathrm{Con}}
\newcommand{\sk}{sk}
\newcommand{\ssk}{\overline{\sk}}

\newcommand{\encoracle}{\cE}
\newcommand{\ct}{ct}

\newcommand{\pop}{N}
\newcommand{\sample}{n}
\newcommand{\length}{\ell}
\newcommand{\dimension}{d}
\newcommand{\dist}{\cD}
\newcommand{\queries}{\length}
\newcommand{\query}{q}
\newcommand{\collusion}{\sample}
\newcommand{\answer}{a}
\newcommand{\rindex}{j}
\newcommand{\errind}[2]{\mathbb{I}\left( #1 \in \{0,1\} \wedge #2 \neq 2#1-1\right)}
\newtheorem{theorem}{Theorem}[section]
\newtheorem{thm}[theorem]{Theorem}

\newtheorem{lem}[theorem]{Lemma}

\newtheorem{claim}[theorem]{Claim}
\newtheorem{remark}[theorem]{Remark}

\newtheorem{prop}[theorem]{Proposition}

\theoremstyle{definition}

\newtheorem{definition}[theorem]{Definition}

\title{Interactive Fingerprinting Codes and the \\ Hardness of Preventing False Discovery}
\author{Thomas Steinke\thanks{Harvard University School of Engineering and Applied Sciences.  Supported by NSF grant CCF-1116616. \newline Email: \href{mailto:tsteinke@seas.harvard.edu}{tsteinke@seas.harvard.edu}.} \and Jonathan Ullman\thanks{Harvard University Center for Research on Computation and Society and Columbia University.  Supported by NSF Grant CNS-1237235 and a Simons Society of Fellows Junior Fellowship.  Email: \href{mailto:jullman@cs.columbia.edu}{jullman@cs.columbia.edu}.}}

\begin{document}
\maketitle

\pagenumbering{gobble}
\begin{abstract}
We show an essentially tight bound on the number of adaptively chosen statistical queries that a computationally efficient algorithm can answer accurately given $n$ samples from an unknown distribution.  A statistical query asks for the expectation of a predicate over the underlying distribution, and an answer to a statistical query is accurate if it is ``close'' to the correct expectation over the distribution.  This question was recently \linebreak studied by Dwork et al.~\cite{DworkFHPRR14}, who showed how to answer $\tilde{\Omega}(n^2)$ queries efficiently, and also by Hardt and Ullman~\cite{HardtU14}, who showed that answering $\tilde{O}(n^3)$ queries is hard.  We close the gap between the two bounds and show that, under a standard hardness \linebreak assumption, there is no computationally efficient algorithm that, given $n$ samples from an unknown distribution, can give valid answers to $O(n^2)$ adaptively chosen statistical queries.  An implication of our results is that computationally efficient algorithms for answering \linebreak arbitrary, adaptively chosen statistical queries may as well be \emph{differentially private}.  

We obtain our results using a new connection between the problem of answering adaptively chosen statistical queries and a combinatorial object called an \emph{interactive fingerprinting code}~\cite{FiatT01}.  In order to optimize our hardness result, we give a new Fourier-analytic approach to analyzing fingerprinting codes that is simpler, more flexible, and yields better parameters than previous constructions.
\end{abstract}

\vfill
\newpage

\tableofcontents

\vfill
\newpage

\pagenumbering{arabic}
\section{Introduction}
Empirical research commonly involves asking multiple ``queries'' on a finite sample drawn from some population(e.g., summary statistics, hypothesis tests, or learning algorithms).  The outcome of a query is deemed significant if it is unlikely to have occurred by chance alone, and a ``false discovery'' occurs if the analyst incorrectly declares an observation significant.  For decades statisticians have been devising methods for preventing false discovery, such as the ``Bonferroni correction''~\cite{Bonferroni36,Dunn61} and the widely used and highly influential method of Benjamini and Hochberg~\cite{BenjaminiH79} for controlling the ``false discovery rate.''  

Nevertheless, false discovery persists across all empirical sciences, and both popular and scientific articles report on an increasing number of invalid research findings.  Typically false discovery is attributed to misuse of statistics.  However, another possible explanation is that methods for preventing false discovery do not address the fact that data analysis is inherently \emph{adaptive}---the choice of queries depends on previous interactions with the data.  The issue of adaptivity was recently investigated in a striking paper by Dwork, Feldman, Hardt, Pitassi, Reingold, and Roth~\cite{DworkFHPRR14} and also by~\cite{HardtU14}.

These two papers formalized the problem of adaptive data analysis in Kearns' \emph{statistical-query (SQ) model}~\cite{Kearns93}. In the SQ model, there is an algorithm called the \emph{oracle} that is given $n$ samples from an unknown distribution ${\cal D}$ over some finite universe $\mathcal{X} = \{0,1\}^d$, where the parameter~$d$ is the dimensionality of the distribution. The oracle must answer \emph{statistical queries} about $\cal D$.  A statistical query $q$ is specified by a predicate 
$p\colon {\cal X}\to \{0,1\}$ and the answer to a statistical query is $$q({\cal D}) = \ex{x\sim{\cal D}}{p(x)}.$$

The oracle's answer $a$ to a query $q$ is \emph{accurate} if $|a-q({\cal D})| \le \alpha$ with high probability (for suitably small $\alpha$).  Importantly, the goal of the oracle is to provide answers that ``generalize'' to the underlying distribution, rather than answers that are specific to the sample. The latter is easy to achieve by outputting the empirical average of the query predicate on the sample.

The analyst makes a sequence of queries $q^1,q^2,\dots,q^k$ to the oracle, which responds with answers $a^1, a^2, \dots, a^k$. In the adaptive setting, the query $q^i$ may depend on the previous queries and answers $q^1,a^1,\dots,q^{i-1},a^{i-1}$ arbitrarily.  We say the oracle is \emph{accurate} given $n$ samples for $k$ adaptively chosen queries if, when given $n$ samples from an arbitrary distribution ${\cal D}$, the oracle accurately responds to any adaptive analyst that makes at most $k$ queries with high probability. A computationally efficient oracle answers each query in time polynomial in $n$ and $d$.\footnote{We assume that the analyst only asks queries that can be evaluated on the sample in polynomial time.}

When the queries are specified \emph{non adaptively} (i.e. independent of previous answers), then the empirical average of each query on the sample is accurate with high probability as long as $k \leq 2^{o(n)}$. However, the situation turns out to be very different when the queries are asked adaptively.  
Using a connection to \emph{differential privacy}~\cite{DworkMNS06}, Dwork et al.~\cite{DworkFHPRR14} showed that there is a computationally efficient oracle that accurately answers $\tilde{\Omega}(n^2)$ many adaptively chosen queries.  They also showed that there is an exponential-time oracle that can answer exponentially in $n$ many queries, and left it open whether this guarantee could be achieved by an efficient oracle.  Unfortunately,~\cite{HardtU14}, building on hardness results for differential privacy~\cite{Ullman13, BunUV14} showed that, assuming the existence of one-way functions, there is no computationally efficient algorithm that answers $\tilde{O}(n^3)$ queries.  Given the importance of preventing false discovery in adaptive data analysis, we would like to know if there is an efficient algorithm that answers as many as $n^3$ queries.

Unfortunately, we show that this is not the case, and prove the following nearly optimal hardness result for preventing false discovery.
\begin{theorem}[Informal] \label{thm:main1}
Assuming the existence of one-way functions, there is no computationally efficient oracle that given $n$ samples is accurate on $O(n^2)$ adaptively chosen queries.
\end{theorem}

Conceptually, our result gives further evidence that there may be an inherent computational barrier to preventing false discovery in interactive data analysis.  It also shows that in the worst case, an efficient oracle for answering adaptively chosen statistical queries may as well be differentially private.  That is, the oracle used in Dwork et al.~\cite{DworkFHPRR14} to answer $\tilde{\Omega}(n^2)$ queries gives the strong guarantee of differential privacy for the sample, and no efficient oracle regardless of privacy can answer significantly more arbitrary, adaptively chosen queries.  It would be interesting to see whether this sort of equivalence holds in more restricted settings.

As in~\cite{HardtU14}, our hardness result applies whenever the dimensionality $d$ of the data grows with the sample size such that $2^d$ is not polynomial in~$n.$\footnote{This is under the stronger, but still standard, assumption that exponentially-hard one-way functions exist.}  This requirement is both mild and necessary.  If $n \gg 2^d$ then the empirical distribution of the $n$ samples will be close to the underlying distribution in statistical distance, so every statistical query can be answered accurately given the sample.  Thus, the dimensionality of the data has a major effect on the hardness of the problem.  In fact, we can prove a nearly optimal \emph{information theoretic} lower bound when the dimensionality of the data is much larger than $n$.

\begin{theorem}[Informal] \label{thm:main2}
There is no oracle (even a computationally unbounded one) that given $n$ samples in dimension $d = O(n^2)$ is accurate on $O(n^2)$ adaptively chosen queries.
\end{theorem}

Our result builds on the techniques of~\cite{HardtU14}, who use \emph{fingerprinting code}~\cite{BonehS98, Tardos03} to prove their hardness result.  In this work, we identify a variant called an \emph{interactive fingerprinting code}~\cite{FiatT01}, which abstracts the technique in~\cite{HardtU14} and gives a more direct way of proving hardness results for adaptive data analysis.  A slightly weaker version of our results can be obtained using the nice recent construction of interactive fingerprinting codes due to Laarhoven et al.~\cite{LaarhovenDRSdW13} as a black box.  However, we give a new analysis of (a close variant of) their code, which is simpler and achieves stronger parameters.
\jnote{We should add a forward reference to some text where we very clearly summarize the improvements.}

Thus, we can summarize the contributions of this work as follows.
\begin{enumerate}
\item We identify \emph{interactive fingerprinting codes} as the key combinatorial object underlying the hardness of preventing false discovery in adaptive environments, analogous to the way in which (non interactive) fingerprinting codes are the key combinatorial object underlying the hardness of differential privacy.
\item We use this connection to prove nearly optimal hardness results for preventing false discovery in interactive data analysis.
\item We give a new Fourier-analytic method for analyzing both interactive and non-interactive fingerprinting codes that we believe is more intuitive, more flexible, and also leads to even stronger hardness results.  In particular, using our analysis we are able to prove that these codes are \emph{optimally robust}\footnote{In this context, optimal robustness means that all of our hardness results apply even when the oracle answers only a $1/2 + \Omega(1)$ fraction of the queries accurately.}~\cite{BunUV14}, which can be used to strengthen the hardness results in~\cite{Ullman13, BunUV14, SteinkeU15}.  Given the importance of fingerprinting codes to adaptive data analysis and privacy, we believe this new analysis will find further applications.
\end{enumerate}

\subsection{Techniques}

The structure of our proof is rather simple, and closely follows the framework in~\cite{HardtU14}. We will design a challenge distribution~${\cal D}$ and a computationally efficient adaptive analyst~${\cal A}$ who knows $\cal D$. If any computationally efficient oracle~${\cal O}$ is given $n$ samples $S=\{x_1,\dots, x_n\}$ drawn from ${\cal D}$, then our analyst~${\cal A}$ can use the answers of $\cal O$ to reconstruct the set $S$.  Using this information, the adversary can construct a query on which $S$ is not representative of $\cal D$.

Our adversary $\cal A$ and the distribution $\cal D$, like that of \cite{HardtU14}, is built from a combinatorial object with a computational ``wrapper.'' The computational wrapper uses queries that cryptographically ``hide'' information from the oracle $\cal O$.  In our work he combinatorial object will be an \emph{interactive fingerprinting code (IFPC)}.  An IFPC is a generalization of a \emph{(standard) fingerprinting code}, which was originally introduced by Boneh and Shaw~\cite{BonehS98} as a way to watermark digital content.

An interactive fingerprinting code $\ifpc$ is an efficient interactive algorithm that defeats any adversary $\fpcadv$ in the following game (with high probability). The adversary $\fpcadv$ picks $S \subset [\pop]$ unknown to $\ifpc$. The goal of $\ifpc$ is to identify $S$ by making $\length$ interactive queries to $\fpcadv$. $\ifpc$ specifies each query by a vector $\FPCcol \in \pmo^\pop$. In response, the adversary $\fpcadv$ must simply output $\answer \in \pmo$ such that $\answer = \FPCcol_i$ for some $i \in [\pop]$. However, the adversary $\fpcadv$ is restricted to only see $\FPCcol_i$ for $i \in S$. At any time, $\ifpc$ may \emph{accuse} some $i \in [\pop]$. If $i \in S$ is accused, then $i$ is removed from $S$ (i.e. $S \leftarrow S \backslash \{i\}$), thereby further restricting $\fpcadv$. If $i \notin S$ is accused, then this is referred to as a \emph{false accusation}. To win, the interactive fingerprinting code $\ifpc$ must accuse all of $S$, without making ``too many'' false accusations.  

In contrast~\cite{HardtU14} use only standard fingerprinting codes.  The difference between interactive and non interactives fingerprinting codes is that a non interactive fingerprinting code must give all $\length$ queries to $\fpcadv$ at once, but is (necessarily) only required to identify one $i \in S$.  The suboptimal parameters achieved by~\cite{HardtU14}, as well as some of the additional technical work, are there result of having to boost non interactive fingerprinting codes to recover all of $S$.  Using this new perspective of interactive fingerprinting codes, the technique of~\cite{HardtU14} can be seen as a construction of an interactive fingerprinting code with length $\length = \tilde{O}(\pop^3)$ by concatenating $\pop$ independent copies of Tardos'~\cite{Tardos03} non interactive fingerprinting code of length $\length = \tilde{O}(\pop^2)$.

However, one can construct more clever and shorter interactive fingerprinting codes.  Specifically, Laarhoven et al.~\cite{LaarhovenDRSdW13} (building on Tardos~\cite{Tardos03}) give a construction that would be suitable for our application with $\length = \tilde{O}(\pop^2)$.  Extending their results, we give a new analysis of their interactive fingerprinting code as well as Tardos' non interactive fingerprinting code that allows us to achieve length $\length = O(\pop^2)$ while still being sufficiently secure for our application.
\begin{theorem}[Informal]
For every $\pop$, there exists an interactive fingerprinting code with $\length = O(\pop^2)$ that, except with negligible probability, makes at most $\pop/1000$ false accusations.
\end{theorem}

This result suffices for the informal statements made above, but our construction is somewhat more general and has additional parameters and security properties, which we detail in Section~\ref{sec:IFPC}.  

\subsection{Applications to Data Privacy}

The adversary used to show hardness of preventing false discovery is effectively carrying out a \emph{reconstruction attack} against the database of samples.  Roughly, if there is an adversary who can reconstruct the set of samples $S$ from the oracle's answers, then the oracle is said to be ``blatantly non-private''---it reveals essentially all of the data it holds, and so cannot guarantee any reasonable notion of privacy to the owners of the data.  Since the seminal work of Dinur and Nissim~\cite{DinurN03}, such reconstruction attacks have been used to establish strong limitations on the accuracy of privacy-preserving oracles.  

Using interactive fingerprinting codes, combined with the framework of~\cite{HardtU14}, we obtain the following results. In both cases, \cite{HardtU14} show similar results, in which our $O(n^2)$ bounds are replaced with $\tilde{O}(n^3)$.

\begin{theorem}[Informal] \label{thm:main3}
Assuming the existence of one-way functions, every computationally \linebreak efficient oracle that, given $n$ samples, is accurate on $O(n^2)$ adaptively chosen queries is blatantly non private.
\end{theorem}

Theorem~\ref{thm:main3} should be compared with the result in~\cite{Ullman13}, which showed that any computationally efficient oracle that, given $n$ samples, is accurate for $\tilde{O}(n^2)$ \emph{non-adaptively chosen} queries cannot satisfy the strong guarantee of ``differential privacy'' \cite{DinurN03,DworkMNS06}.  Theorem \ref{thm:main3} shows that, in the adaptive setting, we can obtain a stronger privacy violation using fewer queries than \cite{Ullman13}.

\begin{theorem}[Informal] \label{thm:main4}
Every (possibly computationally unbounded) oracle that, given $n$ samples in dimension $d = O(n^2)$, is accurate on $O(n^2)$ adaptively chosen queries is blatantly non private.
\end{theorem}

Theorem~\ref{thm:main4} should be compared with the result in~\cite{BunUV14} that showed any (possibly computationally unbounded) oracle that answers a \emph{fixed} family of $\tilde{O}(n^2)$ simple queries in dimension $d = \tilde{O}(n^2)$ cannot satisfy differential privacy.  

In contrast with Theorems \ref{thm:main3} and \ref{thm:main4}, the well-known result of \cite{DworkMNS06} shows that there is an efficient differentially private algorithm that answers $\tilde{\Omega}(n^2)$ adaptively chosen queries. Our results show that, in the adaptive setting, there is a sharp threshold for the number of queries where, below this threshold, the strong notation of differential privacy can be achieved and, above this threshold, even minimal notions of privacy are unachievable.


\subsection{Additional Related Work}
Our work and~\cite{HardtU14} is part of a line of work connecting technology for secure watermarking to lower bounds for private and interactive data analysis tasks.  This connection first appeared in the work of Dwork, Naor, Reingold, Rothblum, and Vadhan~\cite{DworkNRRV09}, who showed that the existence of \emph{traitor-tracing schemes} implies hardness of differential privacy.  Traitor-tracing schemes were introduced by Chor, Fiat, and Naor~\cite{ChorFN94}, also for the problem of watermarking digital content.  The connection between traitor-tracing and differential privacy was strengthened in~\cite{Ullman13}, which introduced the use of fingerprinting codes in the context of differential privacy, and used them to show optimal hardness results for certain settings.  \cite{BunUV14} showed that fingerprinting codes can be used to prove nearly-optimal information-theoretic lower bounds for differential privacy, which established fingerprinting codes as the key information-theoretic object underlying lower bounds in differential privacy.  

Since there introduction by Boneh and Shaw~\cite{BonehS98} there has been extensive work on fingerprinting codes, most of which is beyond the scope of this discussion.  For the standard, non-interactive definition of fingerprinting codes,~\cite{Tardos03} gave an essentially optimal construction, which has been very influential in most of the subsequent work on the topic.  The interactive model of fingerprinting codes was first studied by~\cite{FiatT01} under the name ``dynamic traitor-tracing schemes.''  Formally their results are in a significantly different model and cannot be used to prove hardness of preventing false discovery.  \cite{Tassa05} gave the first construction in the model we use, but achieved suboptimal code length.  Recently Laarhoven, Doumen, Roelse, \v{S}kori\'{c}, and de Wegner~\cite{LaarhovenDRSdW13}, gave a construction with nearly optimal length by generalizing Tardos' code to the interactive setting.  Their construction is quite similar to ours, but our analysis is substantively different and leads to sharper and more general guarantees (and we feel is more intuitive).

In an exciting recent paper,~\cite{DworkFHPRR14} gave the first algorithms for answering arbitrary adaptively chosen statistical queries.  Their algorithms rely on known algorithms for answering statistical queries under differential privacy in a black box manner.  Recently,~\cite{Ullman14} showed how to design differentially private mechanisms for answering exponentially many adaptively chosen queries from the richer class of \emph{convex empirical risk minimization queries}.  By the results of~\cite{DworkFHPRR14}, this algorithm is also a (computationally inefficient) oracle that is accurate for exponentially many adaptively chosen convex empirical risk minimization queries.

\subsection{Organization}

In Section~\ref{sec:IFPC} we define and construct interactive fingerprinting codes, the main technical ingredient we use to establish our results.  In Sections~\ref{sec:falsedisc} and~\ref{sec:nonprivacy} we show how interactive fingerprinting codes can be used to obtain hardness results for preventing false discovery and blatant non privacy, respectively.  The definition of interactive fingerprinting codes is contained in Section~\ref{sec:IFPCdef} and is necessary for Sections~\ref{sec:falsedisc} and~\ref{sec:nonprivacy}, but the remainder of Section~\ref{sec:IFPC} and Sections~\ref{sec:falsedisc} and~\ref{sec:nonprivacy} can be read in either order.

\section{Interactive Fingerprinting Codes}\label{sec:IFPC}

In order to motivate the definition of interactive fingerprinting codes, it will be helpful to review the motivation for standard, non interactive fingerprinting codes. 

Fingerprinting codes were introduced by Boneh and Shaw~\cite{BonehS98} for the problem of watermarking digital content (such as a movie or a piece of software).  Consider a company that distributes some content to $N$ users. Some of the users may illegally distribute copies of the content. To combat this, the company gives each user a unique version of the content by adding distinctive ``watermarks'' to it. Thus, if the company finds an illegal copy, it can be traced back to the user who originally purchased it. Unfortunately, users may be able to remove the watermarks. In particular, a coalition of users may combine their copies in a way that mixes or obfuscates the watermarks. A fingerprinting code ensures that, even if up to $n$ users collude to combine their codewords, an illegal copy can be still be traced to at least one of the users.

Formally, every user $i \in [\pop]$ is given a codeword $(\FPCcol^1_i, \FPCcol^2_i, \dots, \FPCcol^\length_i)  \in \pmo^\length$ by the fingerprinting code, which represents the combination of watermarks in that user's copy. A subset $S \subset [\pop]$ of at most $\collusion$ users can arbitrarily combine their codewords to create a ``pirate codeword'' $\answer = (\answer^1, \answer^2, \dots, \answer^\length) \in \pmo^\length$. The only constraint is so-called \emph{consistency}---for every $\rindex \in [\length]$, if, for every colluding user $i \in S$, we have $\FPCcol^\rindex_i = b$, then $\answer^\rindex=b$. That is to say, if each of the colluding users receives the same watermark, then their combined codeword must also have that watermark. Given $\answer$, the fingerprinting code must be able to trace at least one user $i \in S$.
Tardos \cite{Tardos03} constructed optimal fingerprinting codes with $\ell = O(n^2 \log N)$.

A key drawback of fingerprinting codes is that we can only guarantee that a single user $i \in S$ is traced. This is inherent, as setting the pirate codeword $\answer$ to be the codeword of a single user prevents any other user from being identified. We will see that this can be circumvented by moving to an interactive setting.

Suppose the company is instead distributing a \emph{stream} of content (such as a TV series) to $N$ users---that is, the content is not distributed all at once and the illegal copies are obtained whilst the content is being distributed (e.g. the episodes of the TV series appear on the internet before the next episode is shown). Again, the content is watermarked so that each user receives a unique stream and a subset $S \subset [\pop]$ of at most $\collusion$ users combine their streams and distribute an illegal stream. The company obtains the illegal stream and uses this to trace the colluding users $S$. As soon as the company can identify a colluding user $i \in S$, that user's stream is terminated (e.g. their subscription is cancelled). This process continues until every $i \in S$ has been traced and the distribution of illegal copies ceases.

Another twist on fingerprinting codes is robustness \cite{BunUV14}.  Suppose that the consistency constraint only holds for $(1-\rob) \length$ choices of $\rindex \in [\length]$. That is to say, the colluding users can somehow remove a $\rob$ fraction of the watermarks.  \cite{BunUV14} showed how to modify the Tardos fingerprinting code to be robust to a small constant fraction of inconsistencies. In this work, we show that robustness to any $\rob < 1/2$ fraction of inconsistencies can be achieved.

\subsection{Definition and Existence} \label{sec:IFPCdef}

We are now ready to formally define interactive fingerprinting codes.  To do so we make use of the following game between an adversary $\fpcadv$ and the fingerprinting code $\ifpc$.  Both $\fpcadv$ and $\ifpc$ may be stateful.
\begin{figure}[ht]
\begin{framed}
\begin{algorithmic}
\INDSTATE[0]{$\fpcadv$ selects a subset of users $S^{1} \subseteq [\pop]$ of size $\collusion$, unknown to $\ifpc$.}
\INDSTATE[0]{For $\rindex = 1,\dots,\length$:}
\INDSTATE[1]{$\ifpc$ outputs a column vector $\FPCcol^{\rindex} \in \pmo^{\pop}$.}
\INDSTATE[1]{Let $\FPCcol^{\rindex}_{S^{\rindex}} \in \pmo^{|S^{\rindex}|}$ be the restriction of $\FPCcol^{\rindex}$ to coordinates in $S^{\rindex}$, which is given to $\fpcadv$.}
\INDSTATE[1]{$\fpcadv$ outputs $\answer^{\rindex} \in \pmo$, which is given to $\ifpc$.}
\INDSTATE[1]{$\ifpc$ accuses a (possibly empty) set of users $I^{j} \subseteq [\pop]$.  Let $S^{\rindex+1} = S^{\rindex} \setminus I^{\rindex}$.}
\end{algorithmic}
\end{framed}
\vspace{-6mm}
\caption{$\ifpcgame_{\pop, \collusion, \length}[\fpcadv, \ifpc]$}
\end{figure}
\newcommand{\errors}{\theta}
\newcommand{\falseaccs}{\psi}
For a given execution of $\ifpc$, we let $\FPCmat \in \pmo^{\pop \times \length}$ be the matrix with columns $\FPCcol^{1},\dots,\FPCcol^{\length}$ and let $\answer \in \pmo^{\length}$ be the vector with entries $\answer^{1},\dots,\answer^{\length}$.  We want to construct the fingerprinting code so that, if $\answer$ is consistent, then the tracer succeeds in recovering every user in $S$.  For convenience, we will define the notation $\errors^j$ to be the number of rounds $1,\dots,j$ in which $\answer^j$ is not consistent with $\FPCcol^{j}$.  Formally, for a given execution of $\ifpc$,
$$
\errors^\rindex = \left| \set{1 \leq k \leq \rindex \left|\; \not\exists \; i \in [\pop],\; \answer^{k} = \FPCcol^{k}_{i} \right. } \right|.
$$
Using this notation, $\answer$ is $\rob$-consistent if $\errors^{\length} \leq \rob \length$.  We also define the notation $\falseaccs^\rindex$ to be the number of users in $I^{1},\dots,I^{\rindex}$ who are falsely accused (i.e.~not in the coalition $S^1$).  Formally,
$$
\falseaccs^\rindex = \left| \left(\bigcup_{1 \leq k \leq \rindex} I^{k} \right) \setminus S^1 \right|.
$$
Using this notation, we require $\falseaccs^\length \leq \fpcfalse (\pop-|S^1|)$ - that is, the tracing algorithm does not make too many false accusations. ``Too many'' is defined as more than a $\delta$-fraction of innocent users.

\begin{definition}[Interactive Fingerprinting Codes]
We say that an algorithm $\ifpc$ is an \emph{$\collusion$-collusion-resilient interactive fingerprinting code of length $\length$ for $\pop$ users robust to a $\rob$ fraction of errors with failure probability $\fpcfail$ and false accusation probability $\fpcfalse$} if for every adversary $\fpcadv$, it holds that
$$
\pr{\ifpcgame_{\pop, \collusion, \length}[\fpcadv, \ifpc]}{\left(\errors^{\length} \leq \rob \length\right) \lor \left(\falseaccs^\length > \fpcfalse (\pop-\collusion)\right)} \leq \fpcfail
$$
The length $\length$ may depend on $\pop, \collusion, \rob, \fpcfail, \fpcfalse$.
\end{definition}
The constraint $\falseaccs^\length \leq \fpcfalse \pop$ is called \emph{soundness}---the interactive fingerprinting code should not make (many) false accusations. The constraint $\errors^\length > \rob \length$ is called \emph{completeness}---the interactive fingerprinting code should force the adversary $\fpcadv$ to be inconsistent.  Although it may seem strange that we make no reference to recovering the coalition $S^1$, notice that if $S^{j} \not\eq \emptyset$, then $\pirate$ can easily be consistent.  Therefore, if the pirate cannot be consistent, it must be the case that $S^{j} = \emptyset$ for some $j$, meaning all of $S^1$ has been accused.

In the remainder of this section, we give a construction of interactive fingerprinting codes, and establish the following theorem.

\begin{theorem}[Existence of Interactive Fingerprinting Codes] \label{thm:ifpcthm}
For every $1 \leq \collusion \leq \pop$, $0 \leq \rob  < 1/2$,  and $0 < \fpcfalse \leq 1$, there is a $\collusion$-collusion-resilient interactive fingerprinting code of length $\length$ for $\pop$ users robust to a $\rob$ fraction of errors with failure probability $$\fpcfail \leq \min\{\fpcfalse (\pop-\collusion), 2^{-\Omega(\fpcfalse (\pop-\collusion))}\} + \fpcfalse^{\Omega\left(\irob\collusion\right)}$$ and false accusation probability $\fpcfalse$ for
$$
\length = O\left( \frac{\collusion^2 \log\left(1/\fpcfalse \right)}{\irob^4} \right).
$$
\end{theorem}
We remark on the parameters of our construction and how they relate to the literature.
\begin{remark}
~

\begin{itemize}
\item The expression for the failure probability $\fpcfail$ is a bit mysterious. To interpret it, we fix $\rob = 1/2 - \Omega(1)$ and consider two parameter regimes: $\fpcfalse (\pop-\collusion) \ll 1$ and $\fpcfalse (\pop-\collusion) \gg 1$.

In the traditional parameter regime for fingerprinting codes $\fpcfalse (\pop-\collusion) = \eps' \ll 1$, and so no users are falsely accused.  Then our fingerprinting code has length $O(\sample^2 \log ((\pop-\collusion)/\eps'))$ and a failure probability of $\eps'$. This matches the result of \cite{LaarhovenDRSdW13}.

However, if we are willing to tolerate falsely accusing a small constant fraction of users, then we can set, for example, $\fpcfalse (\pop-\collusion) = .01 \pop$, and our fingerprinting code will have length $O(\sample^2)$ and failure probability $2^{-\Omega(\sample)}$. To our knowledge, such large values of $\fpcfalse$ have not been considered before. It saves a logarithmic factor in our final result.

\item Our construction works for any robustness parameter $\rob < 1/2$. Previously \cite{BunUV14} gave a construction for $\rob=1/75$ in the non-interactive setting.  Previous constructions in the interactive setting do not achieve \emph{any} robustness $\rob > 0$, even for the weaker model of robustness to erasures~\cite{BonehN08}.

\item Our completeness condition differs subtly from previous work. We require that, with high probability, $$\errors^\length = \left| \set{1 \leq k \leq \length \left|\; \not\exists \; i \in [\pop],\; \answer^{k} = \FPCcol^{k}_{i} \right. } \right| > \rob \length,$$ rather than the weaker  condition $$\left| \set{1 \leq k \leq \length \left|\; \not\exists \; i \in S^1,\; \answer^{k} = \FPCcol^{k}_{i} \right. } \right| > \rob \length.$$ While our version is less natural in the watermarking setting, it is important to our application to false dicsovery. Our interactive fingerprinting code ensures that the adversary cannot be consistent with respect to the population, rather than that it cannot be consistent with respect to the sample.
\end{itemize}
\end{remark}


\subsection{The Construction}

Our construction and analysis is based on the optimal (non interactive) fingerprinting codes of Tardos~\cite{Tardos03}, and the robust variant by Bun et al.~\cite{BunUV14}.  The code is essentially the same, but columns are generated and shown to the adversary one at a time, and tracing is modified to identify users interactively.

We begin with some definitions and notation.
For $0 \leq a < b \leq 1$, let $D_{a,b}$ be the distribution with support $(a,b)$ and probability density function $\mu(p) = C_{a,b} / \sqrt{p(1-p)}$, where $C_{a,b}$ is a normalising constant.\footnote{To sample from $D_{a,b}$, first sample $\varphi \in (\sin^{-1}(\sqrt{a}),\sin^{-1}(\sqrt{b}))$ uniformly, then output $\sin^2(\varphi)$ as the sample.} For $\accuracy,\mass \in (0,1/2)$, let $\overline{D}_{\accuracy,\mass}$ be the distribution on $[0,1]$ that returns a sample from $D_{\accuracy,1-\accuracy}$ with probability $1-2\mass$ and $0$ or $1$ each with probability $\mass$.

For $p \in [0,1]$, let $\FPCcol \sim p$ denote that $\FPCcol \in \{\pm 1\}$ is drawn from the distribution with $\pr{}{\FPCcol=1}=p$ and $\pr{}{\FPCcol=-1}=1-p$. Let $\FPCcol_{1 \cdots n} \sim p$ denote that $\FPCcol \in \{\pm 1\}^n$ is drawn from a product distribution in which $\FPCcol_i \sim p$ independently for all $i \in [n]$. 

Define $\phi^p : \{\pm 1\} \to \mathbb{R}$ by $\phi^0(\FPCcol)=\phi^1(\FPCcol)=0$ and, for $p \in (0,1)$, $\phi^p(1)=\sqrt{(1-p)/p}$ and $\phi^p(-1)=-\sqrt{p/(1-p)}$. The function $\phi^p$ is chosen so that $\phi^p(\FPCcol)$ has mean $0$ and variance $1$ when $\FPCcol \sim p$.

\begin{figure}[ht]
\begin{framed}
\begin{algorithmic}
\INDSTATE[0]{Given parameters $1 \leq \collusion \leq \pop$ and  $0 < \fpcfalse, \rob < 1/2$}
\INDSTATE[0]{Set parameters:
\begin{align*}
\accuracy =&  \frac{\irob}{4\collusion} &&\color{DarkRed}{\geq \Omega \left(  \frac{\irob}{\collusion} \right)}\\
\mass =& \frac{3}{8} + \frac{\rob}{4} &&\color{DarkRed}{= \frac12 - \frac{1}{4}\irob} \\
\sigma =& 64 \cdot \left\lceil \frac{6 \pi \collusion}{\irob^2} \right\rceil \cdot  \left\lceil \log_{e}\left(\frac{32}{\fpcfalse}\right) \right\rceil &&\color{DarkRed}{\leq O \left( \frac{\collusion}{\irob^2} \log\left(\frac{1}{\delta} \right) \right)}\\
\length =& \left\lceil \frac{6 \pi \collusion}{\irob^2} \right\rceil \cdot \sigma  &&\color{DarkRed}{\leq O \left(\frac{\collusion^2}{\irob^4} \log\left(\frac{1}{\delta} \right) \right)}
\end{align*}}
\INDSTATE[0]{Let $s^0_i = 0$ for every $i \in [\pop]$.}
\INDSTATE[0]{For $j = 1,\dots,\length$:}
\INDSTATE[1]{Draw $p^j \sim \overline{D_{\accuracy,\mass}}$ and $\FPCcol^{j}_{1\cdots\pop} \sim p^j$.}
\INDSTATE[1]{Issue $\FPCcol^{j} \in \pmo^{\pop}$ as a challenge and receive $\answer^j \in \pmo$ as the response.}
\INDSTATE[1]{For $i \in [\pop]$, let $s^{j}_{i} = s^{j-1}_{i} + \answer^j \cdot \phi^{p^j}(\FPCcol^{j}_i)$.}
\INDSTATE[1]{Accuse $I^j = \set{i \in [\pop] \mid s^j_{i} > \sigma}$.}
\end{algorithmic}
\end{framed}
\vspace{-6mm}
\caption{The interactive fingerprinting code $\ifpc=\ifpc_{\collusion,\pop,\fpcfalse,\rob}$} \label{fig:IFPCconstruction}
\end{figure}

The fingerprinting code $\ifpc$ is defined in Figure \ref{fig:IFPCconstruction}.  In addition to the precise setting of parameters, we have given asymptotic bounds to help follow the analysis.  We now analyze $\ifpc$ and establish Theorem~\ref{thm:ifpcthm}.  The proof of Theorem \ref{thm:ifpcthm} is split into Theorems \ref{thm:Soundness} and \ref{thm:Completeness}. For convenience, define $I = \bigcup_{\rindex \in [\length]} I^\rindex.$

\subsection{Analysis Overview}

Intuitively, the quantity $s_i^\rindex$, which we call the \emph{score} of user $i$, measures the ``correlation'' between the answers $(\answer^1, \cdots, \answer^\rindex)$ of $\fpcadv$ and the $i$-th codeword $(\FPCcol_i^1, \cdots, \FPCcol_i^\rindex)$, using a particular measure of correlation that takes into account the choices $p^1,\dots,p^\rindex$.  If $s_i^\rindex$ ever exceeds the threshold $\sigma$, meaning that the answers are significantly correlated with the $i$-th codeword, then we accuse user $i$.  Thus, our goal is to show two things:  \emph{Soundness}, that the score of an \emph{innocent} user (i.e.~$i \not\in S^1$) never exceeds the threshold, as the answers cannot be correlated with the unknown $i$-th codeword.  And \emph{completeness}, that the score of every \emph{guilty} user (i.e.~$i \in S^1$) will at some point exceed the threshold, meaning that the answers must correlate with the $i$-th codeword for every $i \in S^1$.

\subsubsection{Soundness}

The proof of soundness closely mirrors Tardos' analysis~\cite{Tardos03} of the non-interactive case.  If $i$ is innocent, then, since $\fpcadv$ doesn't see the codeword $(\FPCcol_i^1, \cdots, \FPCcol_i^\rindex)$ of the $i^\text{th}$ user, there cannot be too much correlation.  In this case, one can show that $s_i^\rindex$ is the sum of $j$ independent random variables, each with mean $0$ and variance $1$, where we take the answers $a^1,\dots,a^\rindex$ as fixed and the randomness is over the choice of the unknown codeword.  By analogy to Gaussian random variables, one would expect that $s_i^\rindex \leq \sigma = \Theta( \sqrt{\length \log(1/\delta)})$ with probability at least $1-\delta$.  Formally, the fact that the score in each round is not bounded prevents the use of a Chernoff bound.  But nonetheless, in Section~\ref{sec:Soundness}, we prove soundness using a Chernoff-like tail bound for $s_i^{\rindex}$.

\subsubsection{Completeness}

To prove completeness, we must show that, for guilty users $i \in S^1$, we have $s_i^\rindex > \sigma$ for some $\rindex \in [\length]$ with high probability.  In Sections \ref{sec:BiasedFourierAnalysis} and \ref{sec:Concentration}, we prove that if $\fpcadv$ gives consistent answers in a $1-\rob$ fraction of rounds, then the sum of the scores for each of the guilty users is large.  Specifically, in Theorem \ref{thm:ScoreBound}, we prove that with high probability
\begin{equation} \label{eqn:CorrelationLB} 
\sum_{i \in S^1} s_i^\length \geq \Theta\left(\length\right)
\end{equation}

The constants hidden by the asymptotic notation are set to imply that, for at least one $i \in S^1$, the score $s_i^\length$ is above the threshold $\sigma = \Theta\left(\length/\collusion\right)$.  This step is not too different from the analysis of Tardos and Bun et al.~\cite{Tardos03,BunUV14} for the non-interactive case.
To show that, for \emph{every} $i \in S^1$, we will have $s_i^j > \sigma$ at some point, we must depart from the analysis of non-interactive fingerprinting codes.  If $s_i^\rindex > \sigma$, and user $i$ is accused in round $\rindex$, then the adversary will not see the suffix of codeword $(\FPCcol_i^{\rindex+1}, \cdots, \FPCcol_i^\length)$.  By the same argument that was used to prove soundness, the answers will not be correlated with this suffix, so with high probability the score $s_i^\length$ does not increase much beyond $\sigma$.  Thus,
\begin{equation}\label{eqn:CorrelationUB}
\sum_{i \in S^1} s_i^\length \leq n \cdot O( \sigma ) = \Theta\left(\length\right).
\end{equation}

The hidden constants are set to ensure that Equation \eqref{eqn:CorrelationUB} conflicts with Equation \eqref{eqn:CorrelationLB}. Thus, we can conclude that $\fpcadv$ cannot give consistent answers for a $1-\rob$ fraction of rounds. That is to say, $\fpcadv$ is forced to be inconsistent because all of $S^1$ is accused and eventually $\fpcadv$ sees none of the codewords and is reduced to guessing an answer $\answer^\rindex$.

\subsubsection{Establishing Correlation}

Proving Equation \eqref{eqn:CorrelationLB} is key to the analysis. Our proof thereof combines and simplifies the analyses of \cite{Tardos03} and \cite{BunUV14}.  For this high level overview, we ignore the issue of robustness and fix $\rob = 0$.

First we prove that the correlation bound holds in expectation and then we show that it holds with high probability using an Azuma-like concentration bound.  (Again, as the random variables being summed are not bounded, we are forced to use a more tailored analysis to prove concentration.)  We show that it holds in expectation for each round.  In Proposition \ref{prop:xiExpectation}, we show that the concentration grows in expectation in each round.  For every $\rindex \in [\length]$,
\begin{equation} \label{eqn:ExCorrelation} 
\ex{}{\sum_{i \in S^1} s_i^\rindex - s_i^{\rindex-1}} = \ex{}{\sum_{i \in S^1} \answer^\rindex \cdot \phi^{p^\rindex}(\FPCcol_i^\rindex)} \geq \Omega(1),
\end{equation} 
where the expectations are taken over the randomness of $p^\rindex$, $\FPCcol^\rindex$, and $\answer^\rindex$. 
Equation~\eqref{eqn:ExCorrelation}, combined with a concentration result, implies Equation~\eqref{eqn:CorrelationLB}. 


The intuition behind Equation \eqref{eqn:ExCorrelation} and the choice of $p^\rindex$ is as follows. Consistency guarantees that, if $\FPCcol^\rindex_i=b$ for all $i \in S^1$, then $\answer^\rindex=b$. This is a weak correlation guarantee, but it suffices to ensure correlation between $\answer^\rindex$ and $\sum_{i \in S^1} \FPCcol^\rindex_i$. The affine scaling $\phi^{p^\rindex}$ ensures that $\phi^{p^\rindex}(\FPCcol^\rindex_i)$ has mean zero (i.e. is uncorrelated with a constant) and and unit variance (i.e. has unit correlation with itself). The expectation of $\answer^\rindex \cdot \phi^{p^\rindex}(\FPCcol^\rindex_i)$ can be interpreted as the $i$-th first-order Fourier coefficient of $\answer^\rindex$ as a function of $\FPCcol^\rindex$. To understand first-order Fourier coefficients, consider the ``dictator'' function: Suppose $\answer^\rindex = \FPCcol^\rindex_{i^*}$ for some $i^* \in S^1$ - that is, $\fpcadv$ always outputs the $i^*$-th bit. Then $$\ex{\answer^\rindex,\FPCcol^\rindex, p^\rindex}{\answer^\rindex \sum_{i \in S^1} \phi^{p^\rindex}(\FPCcol^\rindex_i)} = \ex{\FPCcol^\rindex, p^\rindex}{\FPCcol^\rindex_{i^*} \cdot \phi^{p^\rindex}(\FPCcol^\rindex_{i^*})} = \ex{p^\rindex}{2\sqrt{p^\rindex(1-p^\rindex)}} = \Theta(1).$$
This example can be generalised to $\answer^\rindex$ being an arbitrary function of $\FPCcol^\rindex_{S^1}$ using Fourier analysis. This calculation also indicates why we choose the probability density function of $p^\rindex \sim D_{\accuracy,1-\accuracy}$ to be proportional to $1/\sqrt{p(1-p)}$. 

To handle robustness ($\rob>0$) we use the ideas of \cite{BunUV14}.  With probability $2\mass$ each round is a ``special'' constant round---i.e. $\FPCcol^\rindex=(1)^\pop$ or $\FPCcol^\rindex=(-1)^\pop$.  Otherwise it is a ``normal'' round where $\FPCcol^\rindex$ is sampled as before. Intuitively, the adversary $\fpcadv$ cannot distinguish the special rounds from the normal rounds in which $\FPCcol$ happens to be constant. If the adversary gives inconsistent answers on normal rounds, then it must also give inconsistent answers on special rounds. Since there are many more special rounds than normal rounds, this means that a small number of inconsistencies in normal rounds implies a large number of inconsistencies on special rounds. Conversely, inconsistencies are absorbed by the special rounds, so we can assume there are very few inconsistencies in normal rounds.  Thus $\fpcadv$ is forced to behave consistently on the normal rounds and the analysis on these rounds proceeds as before.

\subsection{Proof of Soundness} \label{sec:Soundness}

We first show that no user is falsely accused except with probability $\fpcfalse/2$. This boils down to proving a concentration bound. Then another concentration bound shows that with high probability at most a $\fpcfalse$ fraction of users are falsely accused.

These concentrations bounds are essentially standard. However, we are showing concentration of sums of variables of the form $\phi^p(\FPCcol)$, which may be quite large if $p \approx 0$ or $p \approx 1$. This technical problem prevents us from directly applying standard concentration bounds. Instead we open up the standard proofs and verify the desired concentration. We take the usual approach of bounding the moment generating function and using that to give a tail bound.

\begin{lem} \label{lem:PhiMGF}
For $p \in [\accuracy,1-\accuracy] \cup \{0,1\}$ and $t \in [-\sqrt{\accuracy}/2,\sqrt{\accuracy}/2]$, we have
$$\ex{\FPCcol \sim p}{e^{t \phi^p(\FPCcol)}} \leq e^{t^2}.$$
\end{lem}
\begin{proof}
If $p \in \{0,1\}$, $\phi^p=0$ and the result is trivial.
We have $\ex{\FPCcol \sim p}{\phi^p(\FPCcol)}=0$, $\ex{\FPCcol \sim p}{\phi^p(\FPCcol)^2} = 1$, and, for $\FPCcol \in \{\pm 1\}$, $|\phi^p(\FPCcol)| \leq 1/\sqrt{\accuracy}$. In particular, $|\phi^p(\FPCcol) \cdot t| \leq 1/2$. For $u \in [-1/2,1/2]$, we have $e^u \leq 1 + u + u^2$. Thus $$\ex{\FPCcol \sim p}{e^{t \phi^p(\FPCcol)}} \leq 1 + t\ex{\FPCcol \sim p}{\phi^p(\FPCcol)} + t^2 \ex{\FPCcol \sim p}{\phi^p(\FPCcol)^2} = 1 + t^2 \leq e^{t^2}.$$
\end{proof}

\begin{lem} \label{lem:PhiSum}
Let $p_1 \cdots p_m \in [\accuracy,1-\accuracy] \cup \{0,1\}$ and $\FPCcol_1 \cdots \FPCcol_m$ drawn independently with $\FPCcol_i \sim p_i$. Let $\answer_1 \cdots \answer_m \in [-1,1]$ be fixed. For all $\lambda\geq 0$, we have $$\pr{}{\sum_{i \in [m]} \answer_i \phi^{p_i}(\FPCcol_i) \geq \lambda} \leq e^{-\lambda^2/4m} + e^{ - \sqrt\accuracy \lambda /4}.$$
\end{lem}
\begin{proof}
By Lemma \ref{lem:PhiMGF}, for all $t \in [-\sqrt{\accuracy}/2,\sqrt{\accuracy}/2]$, $$\ex{\FPCcol}{e^{t \sum_{i \in [m]} \answer_i \phi^{p_i}(\FPCcol_i)}} \leq \prod_{i \in [m]} \ex{\FPCcol_i}{e^{t\answer_i\phi^{p_i}(\FPCcol_i)}} \leq e^{t^2m}.$$ By Markov's inequality, $$\pr{}{\sum_{i \in [m]} \answer_i \phi^{p_i}(\FPCcol_i) \geq \lambda} \leq \frac{\ex{}{e^{t \sum_{i \in [m]} \answer_i \phi^{p_i}(\FPCcol_i)}}}{e^{t\lambda}} \leq e^{t^2m-t\lambda}.$$ Set $t = \min \{ \sqrt{\accuracy}/2, \lambda/2m \}$. If $\lambda \in [0,m\sqrt\accuracy]$, then $$\pr{}{\sum_{i \in [m]} \answer_i \phi^{p_i}(\FPCcol_i) \geq \lambda} \leq e^{-\lambda^2/4m}.$$
On the other hand, if $\lambda \geq m\sqrt\accuracy$, then $$\pr{}{\sum_{i \in [m]} \answer_i \phi^{p_i}(\FPCcol_i) \geq \lambda} \leq e^{\accuracy m/4 - \sqrt\accuracy \lambda /2} \leq e^{ - \sqrt\accuracy \lambda /4}.$$ The result is obtained by adding these expressions.
\end{proof}

The following theorem shows how we can beat the union bound for tail bounds on partial sums.
\begin{thm}[Etemadi's Inequality \cite{etemadi}] \label{thm:Etemadi}
Let $X_1 \cdots X_n \in \mathbb{R}$ be independent random variables. For $k \in [n]$, define $S_k = \sum_{i \in [k]} X_i$ to be the $k^\text{th}$ partial sum. Then, for all $\lambda > 0$, $$\pr{}{\max_{k \in [n]} |S_k| > 4 \lambda } \leq 4 \cdot \max_{k \in [n]} \pr{}{|S_k| > \lambda}.$$
\end{thm}

\begin{prop}[Individual Soundness] \label{prop:IndivSoundness}
For all  $i \in [\pop]$, we have $$\pr{}{i \in I \setminus S^1} \leq 8 (e^{-\sigma^2/64\length}+ e^{-\sigma \sqrt{\accuracy} /16})  \leq \fpcfalse/2,$$ where the probability is taken over $\ifpcgame_{\pop, \leq\pop, \length}[\fpcadv,\ifpc_{\pop,\collusion,\fpcfalse,\rob}]$ for an arbitrary $\fpcadv$.
\end{prop}
Here $\ifpcgame_{\pop, \leq \collusion, \length}$ denotes $\ifpcgame_{\pop, \collusion, \length}$ with the constraint $|S^1|=\collusion$ replaced by the constraint $|S^1| \leq \collusion$.
\begin{proof}
Let $i \in [\pop] \setminus S^1$. 
Since the adversary does not see $\FPCcol^\rindex_i$ for any $\rindex \in [\length]$, we may treat the answers of the adversary as fixed and analyse $s_i^\rindex$ as if $\FPCcol^\rindex_i$ was drawn after the actions of the adversary are fixed. Thus, by Lemma \ref{lem:PhiSum}, for every $\rindex \in [\length]$, $$\pr{}{s^\rindex_i>\frac{\sigma}{4}} = \pr{}{\sum_{k \in [\rindex]} \answer^k \phi^{p^k}(\FPCcol^k_i)>\frac{\sigma}{4}} \leq e^{-\sigma^2/64\length}+ e^{-\sigma \sqrt{\accuracy} /16}.$$
Likewise $\pr{}{s^\rindex_i < -\frac{\sigma}{4} }  \leq e^{-\sigma^2/64\length}+ e^{-\sigma \sqrt{\accuracy} /16}.$ Thus, by Theorem \ref{thm:Etemadi}, $$\pr{}{i \in I} \leq \pr{}{\max_{\rindex \in [\length]} |s^\rindex_i| > \sigma} \leq 4 \max_{\rindex \in [\length]} \pr{}{|s_i^\rindex|>\frac{\sigma}{4}} \leq 8 (e^{-\sigma^2/64\length}+ e^{-\sigma \sqrt{\accuracy} /16}) \leq \frac{\fpcfalse}{2}.$$
\end{proof}

\begin{thm}[Soundness] \label{thm:Soundness}
We have
$$\pr{}{|I \setminus S^1| > \fpcfalse (\pop-|S^1|)} \leq \min \left\{ \fpcfalse (\pop-|S^1|), e^{-\fpcfalse (\pop-|S^1|) / 8} \right\},$$
where the probability is taken over $\ifpcgame_{\pop, \leq\pop, \length}[\fpcadv,\ifpc_{\pop,\collusion,\fpcfalse,\rob}]$ for an arbiratry $\fpcadv$.
\end{thm}
\begin{remark}
Interestingly, Theorem \ref{thm:Soundness} does not require $|S^1| \leq \collusion$ -- that is, it holds with respect to $\ifpcgame_{\pop, \leq\pop, \length}[\fpcadv,\ifpc_{\pop,\collusion,\fpcfalse,\rob}]$, rather than $\ifpcgame_{\pop, \collusion, \length}[\fpcadv,\ifpc_{\pop,\collusion,\fpcfalse,\rob}]$. It only requires that $\ifpc$ does not see the codewords of users not in $S^1$. 

This is a useful if we are in a setting where $|S^1|$ is unknown: if $|S^1| > \collusion$, then the interactive fingerprinting code will still not make too many false accusations, even if it fails to identify all of $S^1$.
\end{remark}
\begin{proof}
Let $E_i \in \{0,1\}$ be the indicator of the event $i \in I \backslash S^1$. The $E_i$s for $i \in [\pop]$ are independent (conditioned on the choice of $S^1$ and $p^\rindex$ for $\rindex \in [\length]$). Moreover, by Proposition \ref{prop:IndivSoundness},  $\ex{}{E_i} \leq \fpcfalse/2$ for all $i \in [\pop]$. Thus, by a Chernoff bound, $$\pr{}{|I \backslash S^1| > \fpcfalse (\pop-|S^1|)} = \pr{}{\sum_{i \in [\pop] \backslash S^1} E_i > \fpcfalse (\pop-|S^1|)} \leq e^{-\fpcfalse (\pop-|S^1|) / 8}.$$

If $\fpcfalse < 1/(\pop-|S^1|)$, then this is a very poor bound. Instead we use the fact that the $E_i$s are discrete and Markov's inequality, which amounts to a union bound. For $\fpcfalse (\pop-|S^1|) < 1$, we have 
$$\pr{}{|I \backslash S^1| > \fpcfalse (\pop-|S^1|)} = \pr{}{|I \backslash S^1| \geq 1} \leq \ex{}{\sum_{i \in [\pop] \backslash S^1 } E_i} \leq \frac{\fpcfalse (\pop-|S^1|)}{2} \leq \fpcfalse (\pop-|S^1|).$$
\end{proof}

The following lemma will be useful later.

\begin{lem} \label{lem:HiddenBound}
For $i \in [\pop]$, let $\rindex_i \in [\length+1]$ be the first $\rindex$ such that $i \notin S^\rindex$, where we define $S^{\length+1} = \emptyset$. For any $S \subset [\pop]$, $$\pr{}{\sum_{i \in S} s^\length_i - s^{\rindex_i-1}_i > \lambda} \leq e^{-\lambda^2/4|S|\length} + e^{ - \sqrt\accuracy \lambda /4},$$ where the probability is taken over $\ifpcgame_{\pop, \leq\pop, \length}[\fpcadv,\ifpc_{\pop,\collusion,\fpcfalse,\rob}]$ for an arbitrary $\fpcadv$.
\end{lem}
\begin{proof}
We have $$\sum_{i \in S} s^\length_i - s^{\rindex_i-1}_i = \sum_{i \in S} \sum_{\rindex \in [\length]} \mathbb{I}(\rindex \geq \rindex_i) \answer^\rindex \phi^{p^\rindex}(\FPCcol_i^\rindex).$$ Again, since the adversary doesn't see $\FPCcol_i^\rindex$ for $\rindex \geq \rindex_i$, the random variables $\mathbb{I}(\rindex \geq \rindex_i) \answer^\rindex$ and $\phi^{p^\rindex}(\FPCcol_i^\rindex)$ are independent, so we can view $\mathbb{I}(\rindex \geq \rindex_i) \answer^\rindex \in [-1,1]$ as fixed. Now the result follows from Lemma \ref{lem:PhiSum}.
\end{proof}

\subsection{Proof of Completeness}

To show that the fingerprinting code identifies guilty users we must lower bound the scores $\sum_{i \in S^1} s_i^\length$. First we bound their expectation and then their tails.

\subsubsection{Biased Fourier Analysis} \label{sec:BiasedFourierAnalysis}

For this section, assume that the adversary $\fpcadv$ is always consistent - that is, we have no robustness and $\rob=0$. Robustness will be added in Section \ref{sec:Robustness}. Here we establish that the scores have good expectation, namely $$\ex{}{\sum_{i \in S^1} s_i^\rindex - s_i^{\rindex-1}} \geq \Omega(1)$$ for all $\rindex \in [\length]$.
The score $s^\length_i$ computes the `correlation' between the bits given to user $i$ and the output of the adversary. We must show that that the adversary's consistency constraint implies that there exists some correlation on average. 

In this section we deviate from the proof in \cite{Tardos03}. We use biased Fourier analysis to give a more intuitive proof of the correlation bound.

We have the following lemma and proposition, which relate the correlation $\answer^\rindex \cdot \sum_{i \in S^1} \phi^{p^\rindex}(\FPCcol_i^\rindex)$ to the properties of $\answer^\rindex$ as a function of $p^\rindex$. To interpret these imagine that $f$ represents the adversary $\fpcadv$ with one round viewed in isolation -- the fingerprinting code gives the adversary $\FPCcol^\rindex$ and the adversary responds with $f(\FPCcol^\rindex_{S^\rindex})$.

Firstly, the following lemma gives an interpretation of the correlation value for a fixed $p^\rindex$.

\begin{lem} \label{lem:ExpectationDerivative}
Let $f : \{\pm 1\}^\collusion \to \mathbb{R}$. Define $g : [0,1] \to \mathbb{R}$ by $g(p) = \ex{\FPCcol_{1 \cdots \collusion} \sim p}{f(\FPCcol)}$. For any $p \in (0,1)$, $$\ex{\FPCcol_{1 \cdots \collusion} \sim p}{f(\FPCcol) \cdot \sum_{i \in [\collusion]} \phi^p(\FPCcol_i)} = g'(p)\sqrt{p(1-p)}.$$
\end{lem}
\begin{proof}
For $p \in (0,1)$ and $s \subset [\collusion]$, define $\phi_s^p : \{\pm 1\}^\collusion \to \mathbb{R}$ by $\phi_s^p(\FPCcol) = \prod_{i \in s} \phi^p(\FPCcol_i)$. The functions $\phi_s^p$ form an orthonormal basis with respect to the product distribution with bias $p$ -- that is, $$\forall s,t \subset [n] ~~~~ \ex{\FPCcol_{1 \cdots \collusion} \sim p}{\phi_s^p(\FPCcol) \cdot \phi_t^p(\FPCcol)} = \left\{\begin{array}{cl} 1 & s=t \\ 0 & s \ne t \end{array}\right\}.$$ Thus, for any $p \in (0,1)$, we can write $f$ in terms of these basis functions: $$\forall \FPCcol \in \{\pm 1\}^\collusion ~~~~ f(\FPCcol) = \sum_{s \subset [\collusion]} \tilde{f}^p(s) \phi^p_s(\FPCcol),$$ where $$\forall s \subset [\collusion] ~~~~\tilde{f}^p(s) = \ex{\FPCcol_{1 \cdots \collusion} \sim p}{f(\FPCcol) \phi_s^p(\FPCcol)}.$$
This decomposition is a generalisation of Fourier analysis to biased distributions \cite[\S 8.4]{ODonnell}. For $p,q \in (0,1)$, the expansion of $f$ gives the following expressions for $g(q)$, $g'(q)$ and $g'(p)$.
\begin{align*}
g(q) =& \ex{\FPCcol_{1 \cdots \collusion} \sim q}{f(\FPCcol)}\\
=& \sum_{s \subset [\collusion]} \tilde{f}^p(s) \ex{\FPCcol_{1 \cdots \collusion} \sim q}{\phi^p_s(\FPCcol)}\\
=& \sum_{s \subset [\collusion]} \tilde{f}^p(s) \prod_{i \in s}\ex{\FPCcol \sim q}{\phi^p(\FPCcol)}\\
=& \sum_{s \subset [\collusion]} \tilde{f}^p(s) \left( q \sqrt{\frac{1-p}{p}} - (1-q) \sqrt{\frac{p}{1-p}}\right)^{|s|}.\\
g'(q) =& \sum_{s \subset [\collusion] : s \ne \emptyset} \tilde{f}^p(s) \cdot |s| \cdot \left( q \sqrt{\frac{1-p}{p}} - (1-q) \sqrt{\frac{p}{1-p}}\right)^{|s|-1} \cdot \left( \sqrt{\frac{1-p}{p}} + \sqrt{\frac{p}{1-p}}\right).\\
g'(p) =& \sum_{s \subset [\collusion] : s \ne \emptyset} \tilde{f}^p(s) \cdot |s| \cdot 0^{|s|-1} \cdot \left( \sqrt{\frac{1-p}{p}} + \sqrt{\frac{p}{1-p}}\right)\\
=& \sum_{i \in [\collusion]} \tilde{f}^p(\{i\}) \cdot \left( \sqrt{\frac{1-p}{p}} + \sqrt{\frac{p}{1-p}}\right).
\end{align*}
Note that $\tilde{f}^p(\{i\}) = \ex{\FPCcol_{1 \cdots \collusion} \sim p}{f(\FPCcol) \phi^p(\FPCcol_i)}$ and, hence, $$\ex{\FPCcol_{1 \cdots \collusion} \sim p}{f(\FPCcol) \cdot \sum_{i \in [\collusion]} \phi^p(\FPCcol_i)} = \sum_{i \in [\collusion]} \tilde{f}^p(\{i\}) = \frac{g'(p)}{\sqrt{\frac{1-p}{p}} + \sqrt{\frac{p}{1-p}}} = g'(p) \sqrt{p(1-p)}.$$
\end{proof}

Now we can interpret the correlation for a random $p^\rindex \sim D_{a,b}$.

\begin{prop} \label{prop:ExpectationDifference}
Let $f : \{\pm 1\}^\collusion \to \mathbb{R}$. Define $g : [0,1] \to \mathbb{R}$ by $g(p) = \ex{\FPCcol_{1 \cdots \collusion} \sim p}{f(\FPCcol)}$. For any $0 \leq a < b \leq 1$, $$\ex{p \sim D_{a,b}}{\ex{\FPCcol_{1 \cdots \collusion} \sim p}{f(\FPCcol) \cdot \sum_{i \in [\collusion]} \phi^p(\FPCcol_i)}} = \frac{g(b)-g(a)}{2\sin^{-1}(\sqrt{b}) - 2\sin^{-1}(\sqrt{a})} \geq \frac{g(b)-g(a)}{\pi}.$$
\end{prop}
This effectively follows by integrating Lemma \ref{lem:ExpectationDerivative}.
\begin{proof}
Let $\mu(p) = C_{a,b} / \sqrt{p(1-p)}$ be the probability density function for the distribution $D_{a,b}$ on the interval $(a,b)$. By Lemma \ref{lem:ExpectationDerivative} and the fundamental theorem of calculus, we have
\begin{align*}
\ex{p \sim D_{a,b}}{\ex{\FPCcol_{1 \cdots \collusion} \sim p}{f(\FPCcol) \cdot \sum_{i \in [\collusion]} \phi^p(\FPCcol_i)}} =& \ex{p \sim D_{a,b}}{g'(p)\sqrt{p(1-p)}}\\
=& \int_a^b g'(p) \sqrt{p(1-p)} \mu(p) \mathrm{d}p\\
=& C_{a,b} \int_a^b g'(p) \mathrm{d}p\\
=& C_{a,b} \cdot (g(b)-g(a)).
\end{align*}
It remains to show that $C_{a,b} = \left( 2\sin^{-1}(\sqrt{b}) - 2\sin^{-1}(\sqrt{a}) \right)^{-1} \geq 1/\pi$. This follows from observing that 
$$C_{a,b}^{-1} = \int_a^b \frac{1}{\sqrt{p(1-p)}} \mathrm{d}p = \int_a^b \left( \frac{\mathrm{d}}{\mathrm{d}p} 2 \sin^{-1}(\sqrt{p}) \right) \mathrm{d}p = 2\sin^{-1}(\sqrt{b}) - 2\sin^{-1}(\sqrt{a}) $$ and $$C_{a,b}^{-1} \leq C_{0,1}^{-1} = 2\sin^{-1}(1) - 2\sin^{-1}(0) = \pi.$$
\end{proof}

Now we have a lemma to bring consistency into the picture. If $f$ is consistent, $b \approx 1$, and $a \approx 0$, then $$g(b)-g(a) \approx g(1) - g(0) = f((1)^\collusion) - f((-1)^\collusion) = 1 - (-1) = 2.$$ This gives a lower bound on the correlation for consistent $f$.

\begin{lem} \label{lem:gBounds}
Let $f : \{\pm 1\}^\collusion \to \{\pm 1\}$. Define $g : [0,1] \to [-1,1]$ by $g(p) = \ex{\FPCcol_{1 \cdots \collusion} \sim p}{f(\FPCcol)}$. Suppose $\accuracy \in [0,1]$. Then $|g(1-\accuracy)-g(1)| \leq 2\collusion \accuracy$ and $|g(\accuracy)-g(0)|\leq 2\collusion\accuracy$.
\end{lem}
\begin{proof}
We have $\pr{\FPCcol_{1 \cdots \collusion} \sim 1-\accuracy}{X=(1)^\collusion} = (1-\accuracy)^\collusion$ and 
\begin{align*}
g(1-\accuracy) - g(1) =& f((1)^\collusion) \cdot \pr{\FPCcol_{1 \cdots \collusion} \sim 1-\accuracy}{\FPCcol=(1)^\collusion} + \ex{\FPCcol_{1 \cdots \collusion} \sim p}{f(\FPCcol) | \FPCcol \ne (1)^\collusion} \cdot \pr{\FPCcol_{1 \cdots \collusion} \sim 1-\accuracy}{\FPCcol \ne (1)^\collusion} - g(1)\\
=& g(1) \cdot (1-\accuracy)^\collusion +  \ex{\FPCcol_{1 \cdots \collusion} \sim p}{f(\FPCcol) | \FPCcol \ne (1)^\collusion} \cdot  (1 - (1-\accuracy)^\collusion) - g(1)\\
=& \left(g(1) - \ex{\FPCcol_{1 \cdots \collusion} \sim p}{f(\FPCcol) | \FPCcol \ne (1)^\collusion}\right) \cdot \left((1-\accuracy)^\collusion-1\right).
\end{align*}
Now $\left|g(1) - \ex{\FPCcol_{1 \cdots \collusion} \sim p}{f(\FPCcol) | \FPCcol \ne (1)^\collusion}\right| \leq 2$ and $\left|(1-\accuracy)^\collusion-1\right| \leq \collusion\accuracy$, whence $\left|g(1-\accuracy)-g(1)\right| \leq 2\collusion\accuracy$. The other half of the lemma is symmetric.

\end{proof}

\subsubsection{Robustness} \label{sec:Robustness}

We require the fingerprinting code to be robust to inconsistent answers. We show that the correlation is still good in the presence of inconsistencies.

For $f : \{\pm 1\}^\collusion \to \{\pm 1\}$, define 
a random variable $\xi_{\accuracy,\mass}(f)$ by $$\xi_{\accuracy,\mass}(f) = f(\FPCcol) \cdot \sum_{i \in [\collusion]} \phi^{p}(\FPCcol_i) + \errwt \errind{p}{f(\FPCcol)}, ~~~~p \sim \overline{D_{\accuracy,\mass}}, ~~~~\FPCcol_{1 \cdots \collusion} \sim p,$$ where $\mathbb{I}$ is the indicator function and $\errwt \in (0, 1/2)$ satisfies $\mass \errwt/2 = (1-2\mass)/\pi$ - that is, $$\errwt := \frac{2}{\pi} \frac{1-2\mass}{\mass}.$$

The first term $f(\FPCcol) \cdot \sum_{i \in [\collusion]} \phi^{p}(\FPCcol_i)$ measures the correlation as before. The second term \\$\errwt \errind{p}{f(\FPCcol)}$ measures inconsistencies. We will lower bound the expectation of $\xi_{\accuracy,\mass}(f),$ which amounts to saying ``either there is good correlation \emph{or} there is an inconsistency with good probability.'' Thus either the fingerprinting code is able to accuse users or the adversary is forced to be inconsistent.

The following bounds the expected increase in scores from one round of interaction.

\begin{prop} \label{prop:xiExpectation}
Let $f : \{\pm 1\}^\collusion \to \{\pm 1\}$ and $\accuracy,\mass\in(0,1/2)$. Then $$\ex{}{\xi_{\accuracy,\mass}(f)} \geq \frac{2}{\pi}(1-2\mass)(1-2\collusion\accuracy).$$
\end{prop}
\begin{proof}
Define $g : [0,1] \to [-1,1]$ by $g(p) = \ex{\FPCcol_{1 \cdots \collusion} \sim p}{f(\FPCcol)}$. Now
\begin{align*}
\ex{}{\xi_{\accuracy,\mass}(f)} =& \pr{p \sim \overline{D_{\accuracy,\mass}}}{p=0} \cdot  \errwt \mathbb{I}(f((-1)^\collusion) = 1) +  \pr{p \sim \overline{D_{\accuracy,\mass}}}{p=1} \cdot \errwt \mathbb{I}(f((1)^\collusion) = -1) \\&+\pr{p \sim \overline{D_{\accuracy,\mass}}}{p \in [\accuracy,1-\accuracy]} \cdot \ex{p \sim D_{\accuracy,1-\accuracy}}{\ex{\FPCcol_{1 \cdots \collusion} \sim p}{f(\FPCcol) \cdot \sum_{i \in [\collusion]} \phi^{p}(\FPCcol_i)}}\\
=& \mass \cdot \errwt \left( \mathbb{I}(g(0) = 1) + \mathbb{I}(g(1) = -1)  \right) \\\text{(by Proposition \ref{prop:ExpectationDifference})}~~~&+ (1-2\mass) \cdot \frac{g(1-\accuracy) - g(\accuracy)}{2\sin^{-1}(\sqrt{1-\accuracy}) - 2\sin^{-1}(\sqrt{\accuracy})}\\
\geq& \mass \cdot \errwt \left( \frac{1+g(0)}{2} + \frac{1-g(1)}{2}  \right) + (1-2\mass) \cdot \frac{g(1-\accuracy) - g(\accuracy)}{\pi}\\
=& \frac{1-2\mass}{\pi}\left( 1 + g(0) + 1 - g(1) + g(1-\accuracy) - g(\accuracy) \right)\\
\geq& \frac{1-2\mass}{\pi} \left( 2 - |g(\accuracy)-g(0)| -|g(1-\accuracy)-g(1)| \right)\\
\text{(by Lemma \ref{lem:gBounds})}~~~\geq& \frac{1-2\mass}{\pi}(2-4\collusion\accuracy).
\end{align*}
\end{proof}

\subsubsection{Concentration} \label{sec:Concentration}

So far we have shown that the fingerprinting code achieves good correlation or the adversary is not consistent \emph{in expectation}. However, we need this to hold with high probability. 
Thus we now show that sums of $\xi_{\accuracy,\mass}(f)$ variables concentrate around their expectation.

Again, the proofs in this section are standard. However, the $\xi_{\accuracy,\mass}(f)$ variables can be quite unwieldy and we are thus unable to apply standard results directly.  So instead we must open the proofs and verify that the concentration bounds hold. We proceed by bounding the moment generating function of $\xi_{\accuracy,\mass}(f)$ and then proving an Azuma-like concentration inequality. These calculations are not novel or insightful.

\newcommand{\pell}{j}
\begin{prop} \label{prop:xiConcentration}
Let $f : \{\pm 1\}^\collusion \to \{\pm 1\}$, $\accuracy \in(0,1/2)$, $\mass \in [1/4,1/2)$, and $t \in [-\sqrt{\accuracy}/8,\sqrt{\accuracy}/8]$. Then
$$\ex{}{e^{t(\xi_{\accuracy,\mass}(f)-\ex{}{\xi_{\accuracy,\mass}(f)})}} \leq e^{Ct^2},$$ where $C = \frac{64 e^{\collusion\accuracy/4}}{\accuracy}$.
\end{prop}
\begin{proof}
We have
$$\xi_{\accuracy,\mass}(f) = f(\FPCcol) \cdot \sum_{i \in [\collusion]} \phi^{p}(\FPCcol_i) + \errwt \errind{p}{f(\FPCcol)}, ~~~~p \sim \overline{D_{\accuracy,\mass}}, ~~~~\FPCcol_{1 \cdots \collusion} \sim p.$$ Let $Y= \sum_{i \in [\collusion]} \phi^{p}(\FPCcol_i)$.
By Lemma \ref{lem:PhiMGF} and independence, $$\ex{}{e^{tY}}=\ex{\FPCcol_{1 \cdots \collusion} \sim p}{e^{t \sum_{i \in [\collusion]} \phi^p(\FPCcol_i)}} = \left( \ex{\FPCcol \sim p}{e^{t \phi^p(\FPCcol)}} \right)^\collusion \leq e^{t^2 \collusion}$$ for $t \in [-\sqrt{\accuracy}/2,\sqrt{\accuracy}/2]$. 
Pick $t \in \{\pm \sqrt{\accuracy}/2\}$ such that $$\sum_{k=0}^\infty \frac{t^{2k+1}}{(2k+1)!} \ex{}{Y^{2k+1}} \geq 0.$$ Then by dropping positive terms, for all $\pell \geq 1$, $$0 \leq \ex{}{Y^{2\pell}} \leq \frac{(2\pell)!}{t^{2\pell}} \sum_{k=0}^\infty \frac{t^k}{k!} \ex{}{Y^k} =  \frac{(2\pell)!}{t^{2\pell}}\ex{}{e^{tY}} \leq  \frac{(2\pell)!}{t^{2\pell}}e^{\collusion t^2} = \frac{4^\pell (2\pell)!}{\accuracy^\pell} e^{\collusion\accuracy/4}.$$
Thus we have bounded the even moments of $Y$. By Cauchy-Schwartz, for $k=2\pell+1 \geq 3$, $$\ex{}{|Y|^k} \leq \sqrt{\ex{}{Y^{2\pell}} \cdot \ex{}{Y^{2\pell+2}}} \leq \sqrt{\frac{4^\pell (2\pell)!}{\accuracy^\pell} e^{\collusion\accuracy/4} \cdot \frac{4^{\pell+1} (2\pell+2)!}{\accuracy^{\pell+1}} e^{\collusion\accuracy/4}} = \frac{2^{k}k!}{\accuracy^{k/2}}e^{\collusion\accuracy/4} \sqrt{\frac{k+1}{k}}.$$ Since $|f(\FPCcol)|\leq 1$, we have $\ex{}{|f(\FPCcol) \cdot Y|^k} \leq \ex{}{|Y|^k} \leq 2^{k+1} k! e^{\collusion\accuracy/4}/\accuracy^{k/2}$ for all $k \geq 2$. Since $\mass\in  [1/4, 1/2)$, we have $\errwt = (2/\pi)(1-2\mass)/\mass \in (0, 1)$. Hence $\ex{}{|\errwt \errind{p}{f(\FPCcol)}|^k} \leq 1$ for all $k$. The map $u \mapsto |u|^k$ is convex for all $k \geq 2$, thus $|(x+y)/2|^k \leq (|x|^k+|y|^k)/2$ for all $k\geq 2$ and $x,y \in \mathbb{R}$. Combining these three facts, we have $$\ex{}{|\xi_{\accuracy,\mass}(f)|^k} \leq 2^{k-1}\ex{}{|f(\FPCcol)\cdot Y|^k + |\errwt \mathbb{I}(f(\FPCcol) \ne f^{*}(\FPCcol))|^k} \leq \frac{2^{2k} k! e^{\collusion\accuracy/4}}{\accuracy^{k/2}}+2^{k-1} \leq \frac{2^{2k+1} k! e^{\collusion\accuracy/4}}{\accuracy^{k/2}}.$$
For $t \in [-\sqrt{\accuracy}/8,\sqrt{\accuracy}/8]$, we have 
\begin{align*}
\ex{}{e^{t\xi_{\accuracy,\mass}(f)}} \leq& 1 + t\ex{}{\xi_{\accuracy,\mass}(f)} + \sum_{k=2}^\infty \frac{|t|^k}{k!} \ex{}{|\xi_{\accuracy,\mass}(f)|^k}\\
\leq& 1 + t\ex{}{\xi_{\accuracy,\mass}(f)} + \sum_{k=2}^\infty \frac{|t|^k}{k!} \frac{2^{2k+1} k! e^{\collusion\accuracy/4}}{\accuracy^{k/2}}\\
=& 1 + t\ex{}{\xi_{\accuracy,\mass}(f)} + 2 e^{\collusion\accuracy/4} \sum_{k=2}^\infty \left(\frac{4|t|}{\sqrt{\accuracy}}\right)^k\\
\leq& 1 + t\ex{}{\xi_{\accuracy,\mass}(f)} + 2 e^{\collusion\accuracy/4} \sum_{k=2}^\infty \left(\frac{4|t|}{\sqrt{\accuracy}}\right)^2 2^{-(k-2)}\\
=& 1 + t\ex{}{\xi_{\accuracy,\mass}(f)} + \frac{64 e^{\collusion\accuracy/4}}{\accuracy} t^2\\
\leq&e^{t\ex{}{\xi_{\accuracy,\mass}(f)} + Ct^2}
\end{align*}
\end{proof}

\begin{thm}[Azuma-Doob Inequality] \label{thm:AzumaDoob}
Let $X_1 \cdots X_m \in \mathbb{R}$, $\mu_1 \cdots \mu_m \mathbb{R}$ and $\mathcal{U}_0 \cdots  \mathcal{U}_m \in \Omega$ be random variables such that, for all $i \in [m]$,
\begin{itemize}
\item $X_i$ is determined by $\mathcal{U}_i$, 
\item $\mu_i$ is determined by $\mathcal{U}_{i-1}$, and
\item $\mathcal{U}_{i-1}$ is determined by $\mathcal{U}_i$. 
\end{itemize}
Suppose that, for all $i \in [m]$, $u \in \Omega$, and $t \in [-c,c]$, $$\ex{}{e^{t(X_i-\mu_i)} \mid \mathcal{U}_{i-1}=u} \leq e^{Ct^2}.$$
If $\lambda \in [0,2Cmc]$, then $$\pr{}{\left|\sum_{i \in [m]} (X_i-\mu_i) \right| \geq \lambda} \leq 2 e^{-\lambda^2/4Cm}.$$
If $\lambda \geq 2Cmc$, then $$\pr{}{\left|\sum_{i \in [m]} (X_i-\mu_i) \right| \geq \lambda} \leq 2 e^{mCc^2-c\lambda} \leq 2 e^{-c\lambda/2}.$$
\end{thm}
\begin{proof}
First we show by induction on $k \in [m]$ that, for all $u \in \Omega$ and $t \in [-c,c]$, $$\ex{}{e^{t\sum_{i = m-k+1}^m (X_i-\mu_i)} \mid \mathcal{U}_{m-k}=u} \leq e^{k \cdot Ct^2}.$$ This clearly holds for $k=1$, as this is our supposition for $i=m$. Now suppose this holds for some $k \in [m-1]$. For $u \in \Omega$ and $t \in [-c,c]$, we have
\begin{align*}
\ex{}{e^{t\sum_{i = m-k}^m (X_i-\mu_i)} \mid \mathcal{U}_{m-(k+1)}=u} =& \sum_{v \in \Omega} \pr{}{\mathcal{U}_{m-k}=v \mid \mathcal{U}_{m-k-1}=u} \ex{}{e^{t\sum_{i = m-k}^m (X_i-\mu_i)} \mid \mathcal{U}_{m-k}=v}\\
=& \sum_{v \in \Omega} \pr{}{v \mid u} \ex{}{e^{t(X_{m-k}-\mu_{m-k})} e^{t\sum_{i = m-k+1}^m (X_i-\mu_i)} \mid v}\\
&\text{(using shorthand $v \equiv \mathcal{U}_{m-k}=v$ and $u \equiv \mathcal{U}_{m-k-1}=u$)}\\
=& \sum_{v \in \Omega} \pr{}{v \mid u} \ex{}{e^{t(X_{m-k}-\mu_{m-k})}\mid v} \ex{}{e^{t\sum_{i = m-k+1}^m (X_i-\mu_i)} \mid v}\\
&\text{(since $U_{m-k}=v$ determines $X_{m-k}$ and $\mu_{m-k}$)}\\
\leq& \sum_{v \in \Omega} \pr{}{v \mid u} \ex{}{e^{t(X_{m-k}-\mu_{m-k})}\mid v} e^{k \cdot Ct^2}\\
&\text{(by the induction hypothesis)}\\
=& \ex{}{e^{t(X_{m-k}-\mu_{m-k})}\mid u} e^{k \cdot Ct^2}\\
\leq& e^{ Ct^2} e^{k \cdot Ct^2}\\
&\text{(by our supposition for $i=m-k$)}\\
=& e^{(k+1) \cdot C t^2}.
\end{align*}
Thus, for all $t \in [-c,c]$, we have $$\ex{}{e^{t\sum_{i = 1}^m (X_i-\mu_i)}} \leq e^{m \cdot Ct^2}.$$
By Markov's inequality we have
$$\pr{}{\sum_{i \in [m]} (X_i-\mu_i) \geq \lambda} \leq \frac{\ex{}{e^{t \sum_{i \in [m]} (X_i-\mu_i)}}}{ e^{t\lambda}} \leq e^{mCt^2-t\lambda}$$ and
$$\pr{}{\sum_{i \in [m]} (X_i-\mu_i) \leq -\lambda} \leq \frac{\ex{}{e^{-t \sum_{i \in [m]} (X_i-\mu_i)}}}{ e^{(-t)(-\lambda)}} \leq e^{mCt^2-t\lambda}$$ for all $t \in [0,c]$ and $\lambda > 0$. Set $t = \min\{c,\lambda/2mC\}$ to obtain the result.
\end{proof}

\subsubsection{Bounding the Score} \label{sec:BoundingScore}

Now we can finally show that the scores are large with high probability.

\begin{thm}[Correlation Lower Bound] \label{thm:ScoreBound}
At the end of $\ifpcgame_{\pop, \collusion, \length}[\fpcadv,\ifpc_{\pop,\collusion,\fpcfalse,\rob}]$ for arbitrary $\fpcadv$, we have, for any $\lambda \in [0, 17.5 \length / \sqrt{\accuracy}]$, $$\errwt \err^\length + \sum_{i \in S^1} s^\length_i \geq \frac{2}{\pi}(1-2\mass)(1-2\collusion\accuracy)\length - \lambda$$ with probability at least $1-2e^{-\frac{\lambda^2 \accuracy}{280 \length}}$.
\end{thm}
\begin{proof}
Since the adversary $\fpcadv$ is computationally unbounded and arbitrary, we may assume it is deterministic. We may also assume $\collusion=|S^1|$ and that the adversary is able to see $\FPCcol^\rindex_{S^1}$ at each round. (This only gives the adversary more power.)

This means that for each $\rindex \in [\length]$ we can define a function $f^\rindex : \{\pm 1\}^\collusion \to \{\pm 1\}$ that only depends on the interaction up to round $\rindex-1$ (i.e. is a function of the state of $\fpcadv$ before it receives $\FPCcol^\rindex$) and satisfies $f^\rindex(\FPCcol^\rindex_{S^\rindex})=\answer^\rindex$.
For $\rindex \in [\length]$, define $$X_\rindex := \errwt \cdot \errind{p^\rindex}{f^\rindex(\FPCcol^\rindex_{S^1})} + f^\rindex(\FPCcol^\rindex_{S^1}) \cdot \sum_{i \in S^1} \phi^{p^\rindex}(\FPCcol^\rindex_i) \sim \xi_{\accuracy,\mass}(f^\rindex),$$ where $\sim$ denotes having the same distribution. We have $$\errwt \cdot (\err^\rindex-\err^{\rindex-1}) + \sum_{i \in S^1} (s^\rindex_i-s^{\rindex-1}_i) \leq X_\rindex$$ and $$\errwt \err^\length + \sum_{i \in S^1} s^\length_i \leq \sum_{\rindex \in [\length]} X_\rindex  \sim \sum_{\rindex \in [\length]} \xi_{\accuracy,\mass}(f^\rindex).$$
Now we can apply the above lemmas to bound the expectation and tail of this random variable. 

Firstly, Proposition \ref{prop:xiExpectation} shows that $$\mu_\rindex :=\ex{}{X_\rindex} = \ex{}{\xi_{\accuracy,\mass}(f^\rindex)} \geq \frac{2}{\pi}(1-2\mass)(1-2\collusion\accuracy)$$ for all $f^\rindex$. Moreover, by Proposition \ref{prop:xiConcentration}, $$\ex{}{e^{t(X^\rindex - \mu_\rindex)} } = \ex{}{e^{t(\xi_{\accuracy,\mass}(f^\rindex) - \ex{}{\xi_{\accuracy,\mass}(f^\rindex)})}} \leq e^{Ct^2}$$ for all $t \in [-\sqrt\accuracy/8,\sqrt\accuracy/8]$, where $C =  70/\accuracy \geq 64e^{\collusion\accuracy/4}/\accuracy$, as $\accuracy \leq 1/4\collusion$.

Define $\mathcal{U}_\rindex = (f^1,p^1,\FPCcol^1, \cdots, f^\rindex, p^\rindex,\FPCcol^\rindex,f^{\rindex+1})$ for $\rindex \in [\length] \cup \{0\}$. Now $X_1 \cdots X_\length$, $\mu_1 \cdots \mu_\length$, and $\mathcal{U}_0, \cdots, \mathcal{U}_\length$ satisfy the hypotheses of Theorem \ref{thm:AzumaDoob} with $C = 70/\accuracy$, $c=\sqrt\accuracy/8$, and $m=\length$. 

For $\lambda \in[0, 2Cmc] = [0,17.5 \length/\sqrt\accuracy]$, we have $$\pr{}{\sum_{\rindex \in [\length]} X_\rindex \leq \frac{2}{\pi}(1-2\mass)(1-2\collusion\accuracy) \length - \lambda } \leq \pr{}{\left|\sum_{i \in [m]} (X_i-\mu_i) \right| \geq \lambda} \leq 2e^{-\lambda^2/4Cm} \leq 2 e^{-\frac{\lambda^2 \accuracy}{280\length}},$$ as required.
\end{proof}

However, we can also prove that the scores are small with high probability. This follows from the fact that users with large scores are accused and therefore no user's score can be too large:

\begin{lem} \label{lem:TotalScoreBound}
For all $\lambda>0$,
$$\pr{}{\sum_{i \in S^1} s_i^\length > \lambda + \collusion\sigma + \frac{\collusion}{\sqrt{\accuracy}}} \leq e^{-\lambda^2/4\collusion\length} + e^{ - \sqrt\accuracy \lambda /4},$$
where the probability is taken over $\ifpcgame_{\pop, \collusion, \length}[\fpcadv,\ifpc_{\pop,\collusion,\fpcfalse,\rob}]$ for an arbitrary $\fpcadv$.
\end{lem}
We will set $\lambda = \sigma$ and, since $1/\sqrt{\accuracy} \leq \sigma$, we get that $\sum_{i \in S^1} s_i^\length \leq 3 \sigma n$ with high probability.
\begin{proof}
Let $\rindex_i \in [\length+1]$ be as in Lemma \ref{lem:HiddenBound} -- that is, $i \notin S^{\rindex_i}$ and $i \in S^{\rindex_i-1}$, where we define $S^{\length+1}=\emptyset$ and $S^0 = [\pop]$. By the definition of $\rindex_i$, $s^\rindex$, and $S^\rindex$, we have $s^{\rindex_i-2}_i\leq \sigma$ for all $i \in S^1$, as otherwise $i \in I^{\rindex_i-2}$ and therefore $i \notin S^{\rindex_i-1} = S^{\rindex_i-2} \backslash I^{\rindex_i-2}$. If $i \in S^1$, then $\rindex_i=1$ and $s^{\rindex_i-1}_i = 0$. Thus $$\sum_{i \in S^1} s_i^{\rindex_i-1} = \sum_{i \in S^1} s_i^{\rindex_i-2} + \answer^{\rindex_i-1} \phi^{p^{\rindex_i-1}}(\FPCcol_i^{\rindex_i-1}) \leq \sum_{i \in S^1} \sigma + \frac{1}{\sqrt{\accuracy}} \leq \collusion\sigma + \frac{\collusion}{\sqrt{\accuracy}}.$$

By  Lemma \ref{lem:HiddenBound}, $$\pr{}{\sum_{i \in S^1} s^\length_i - s^{\rindex_i-1}_i > \lambda} \leq e^{-\lambda^2/4\collusion\length} + e^{ - \sqrt\accuracy \lambda /4}.$$ The lemma follows.
\end{proof}

Now we show that the conflicting bounds of Theorem \ref{thm:ScoreBound} and Lemma \ref{lem:TotalScoreBound} imply completeness - that is, the adversary $\fpcadv$ cannot be consistent.

\begin{thm}[Completeness] \label{thm:Completeness}
At the end of $\ifpcgame_{\pop, \collusion, \length}[\fpcadv,\ifpc_{\pop,\collusion,\fpcfalse,\rob}]$ for an arbitrary $\fpcadv$, we have $\err^\length > \rob \length$ with probability at least $1-\fpcfalse^{\frac12 \irob \collusion}$, assuming $\irob\collusion \geq 1$.
\end{thm}
\begin{proof}
Suppose for the sake of contradiction that $\err^\length \leq \rob \length$.
By Lemma \ref{lem:TotalScoreBound}, $\sum_{i \in S^1} s_i^\length \leq \lambda + \collusion\sigma + \frac{\collusion}{\sqrt{\accuracy}}$ with probability at least $1-e^{-\lambda^2/4\collusion \length} - e^{ - \sqrt\accuracy \lambda /4}$. Set $\lambda = \collusion \sigma \geq \frac{\collusion}{\sqrt{\accuracy}}$. Now we assume $$\sum_{i \in S^1} s_i^\length \leq 3\collusion\sigma,$$ which holds with probability at least $1-e^{-\collusion\sigma^2/4\length} - e^{ - \sqrt\accuracy \collusion \sigma /4}$. Then \begin{equation}\label{eqn:ScoreUB}\errwt \err^\length + \sum_{i \in S^1} s^\length_i \leq  \errwt \rob \length + 3\collusion\sigma.\end{equation}

By Theorem \ref{thm:ScoreBound}, with probabilty at least $1-2e^{-\frac{\lambda^2 \accuracy}{280\length}}$, \begin{equation}\label{eqn:ScoreLB} \errwt \err^\length + \sum_{i \in S^1} s^\length_i \geq  \frac{2}{\pi}(1-2\mass)(1-2\collusion\accuracy)\length - \lambda\end{equation} for all $\lambda \in [0, 17.5 \length / \sqrt{\accuracy}]$. Set $\lambda = \irob^2 \length / 2\pi$ and assme Equation \eqref{eqn:ScoreLB} also holds.

Combining Equations \eqref{eqn:ScoreUB} and \eqref{eqn:ScoreLB} gives \begin{equation}\label{eqn:Contradiction} \frac{2}{\pi}(1-2\mass)(1-2\collusion\accuracy)\length - \frac{\irob^2}{2\pi} \length \leq \errwt \rob \length + 3\collusion\sigma.\end{equation} We claim this is a contradiction, which then holds with high probability, thus proving the theorem.

Rearranging Equation \eqref{eqn:Contradiction} gives \begin{equation}\label{eqn:Contradiction2}\frac{2}{\pi}(1-2\mass)(1-2\collusion\accuracy) \leq  \frac{\irob^2}{2\pi} + \errwt \rob + \frac{3\collusion\sigma}{\length}.\end{equation}
Our setting of parameters gives $$2\collusion \accuracy \leq \frac{\irob}{2} ~~~~~\text{and}~~~~~ \frac{3 n \sigma}{\length} \leq \frac{\irob^2}{2\pi}.$$ 
Substituting these into Equation \eqref{eqn:Contradiction2} gives \begin{equation}\label{eqn:Contradiction3} \frac{2}{\pi}(1-2\mass)\left(1 - \frac12 \irob\right) \leq \frac{\irob^2}{\pi} + \errwt\rob.\end{equation}
Now we use $1 - 2\mass = \frac12 \irob$ and $\errwt = \frac{2}{\pi} \frac{1-2\mass}{\mass} = \frac{\irob}{\pi \mass}$ to derive a contradiction from Equation \eqref{eqn:Contradiction3}:
\begin{align*}
\frac{\irob}{\pi}\left(1 - \frac12 \irob\right) \leq& \frac{\irob^2}{\pi} + \frac{\irob}{\pi \mass} \rob,\\
1 - \frac12 \irob\leq& \irob + \frac{\rob}{\mass},\\
\mass \left(1 - \frac{3}{2} \irob \right) \leq& \rob.\\
\end{align*}
Since $\mass = \frac12 -\frac14\irob$, we have $$\mass \left(1 - \frac{3}{2} \irob \right)  = \frac12 \left(1 -  \frac12 \irob \right) \left(1 - \frac{3}{2} \irob \right) > \frac12 \left( 1 - 2 \irob \right).$$ And $$\rob = \frac12\left(1 - 2\irob\right).$$
This gives a contradiction.
The total failure probability is bounded by $$e^{-\collusion\sigma^2/4\length} + e^{ - \sqrt\accuracy \collusion \sigma /4} + 2 e^{-\lambda^2 \accuracy / 280 \length} \leq \left(\frac{\fpcfalse}{32}\right)^{16\collusion} + \left(\frac{\fpcfalse}{32}\right)^{4\collusion} + 2 \left(\frac{\fpcfalse}{32}\right)^{\frac12 \irob\collusion} \leq \fpcfalse^{\frac12 \irob\collusion},$$
assuming $\irob\collusion \geq 1$.
\end{proof}

\subsection{Non-Interactive Fingerprinting Codes}

Our construction and analysis also gives a construction of traditional non-interactive fingerprinting codes. First we give a formal definition of a fingerprinting code.

\begin{definition}[(Non-Interactive) Fingerprinting Codes]
A $\collusion$-collusion resilient \emph{(non-interactive) fingerprinting code} of length $\length$ for $\pop$ users robust to a $\rob$ fraction of errors with failure probability $\fpcfail$ and false accusation probability $\fpcfalse$ is a pair of random variables $\FPCmat \in \{\pm 1\}^{\pop \times \length}$ and $\trace : \{\pm 1\}^\length \to 2^{[\pop]}$ such that the following holds. For all adversaries $\fpcadv : \{\pm 1\}^{\collusion \times \length} \to \{\pm 1\}^\length$ and $S \subset [\pop]$ with $|S|=\collusion$, $$\pr{\FPCmat,\trace,\fpcadv}{\left(  \left|\left\{ 1 \leq \rindex \leq \length : \not\exists i \in [\pop] ~~ \fpcadv(\FPCmat_S)^\rindex = \FPCcol^\rindex_i \right\}\right| \leq \rob \length\right) \wedge \left( \trace( \fpcadv(\FPCmat_S) ) = \emptyset \right)} \leq \fpcfail$$ and $$\pr{\FPCmat,\trace,\fpcadv}{\left|\trace(\fpcadv(\FPCmat_S)) \cap \left( [\pop]\backslash S \right) \right| > \fpcfalse (\pop-\collusion)} \leq \fpcfail,$$ where $\FPCmat_S \in \{\pm 1\}^{\collusion \times \length}$ contains the rows of $\FPCmat$ given by $S$.
\end{definition}


Our construction and analysis is readily adapted to the non-interactive setting. We obtain the following theorem.

\begin{theorem}[Existence of Non-Interactive Fingerprinting Codes] \label{thm:fpcthm}
For every $1 \leq \collusion \leq \pop$, $0 \leq \rob  < 1/2$,  and $0 < \fpcfalse \leq 1$, there is a $\collusion$-collusion-resilient (non-interactive) fingerprinting code of length $\length$ for $\pop$ users robust to a $\rob$ fraction of errors with failure probability $$\fpcfail \leq \min\{\fpcfalse (\pop-\collusion), 2^{-\Omega(\fpcfalse (\pop-\collusion))}\} + \fpcfalse^{\Omega\left(\irob\collusion\right)}$$ and false accusation probability $\fpcfalse$ for
$$
\length = O\left( \frac{\collusion^2 \log\left(1/\fpcfalse \right)}{\irob^4} \right).
$$
\end{theorem}

\section{Hardness of False Discovery} \label{sec:falsedisc}

In this section we prove our main result - that answering $O(n^2)$ adaptive queries given $n$ samples is hard. But first we must formally define the model in which we are working.

\subsection{The Statistical Query Model}

Given a distribution $\dist$ over $\bits^{\dimension}$, we would like to answer \emph{statistical queries} about $\dist$.  A statistical query on $\bits^{\dimension}$ is specified by a function $\query \from \bits^{\dimension} \to [-1,1]$ and (abusing notation) is defined to be
$$
\query(\dist) = \ex{x \getsr \dist}{\query(x)}.
$$

Our goal is to design an \emph{oracle} $\oracle$ that answers statistical queries on $\dist$ using only iid samples $x_{1},\dots,x_{\sample} \getsr \dist$.  Our focus is the case where the queries are chosen adaptively and adversarially.

Specifically, $\oracle$ is a stateful algorithm that holds a collection of samples $x_1,\dots,x_{\sample} \in \bits^{\dimension}$, takes a statistical query $\query$ as input, and returns a real-valued answer $a \in [-1,1]$.  We require that when $x_1,\dots,x_{\sample}$ are iid samples from $\dist$, the answer $a$ is close to $\query(\dist)$, and moreover that this condition holds for every query in an adaptively chosen sequence $\query^{1},\dots, \query^{\queries}$.  Formally, we define the following game between an $\oracle$ and a stateful adversary $\accadv$.
\begin{figure}[ht]
\begin{framed}
\begin{algorithmic}
\STATE{$\accadv$ chooses a distribution $\dist$ over $\bits^{\dimension}$.}
\STATE{Sample $x_1,\dots,x_{\sample} \getsr \dist$, let $x = (x_1,\dots,x_{\sample})$.}
\STATE{For $j = 1,\dots,\queries$}
\INDSTATE[1]{$\accadv$ outputs a query $\query^j$.}
\INDSTATE[1]{$\oracle(x, \query^{j})$ outputs $a^{j}$.}
\INDSTATE[1]{(As $\accadv$ and $\oracle$ are stateful, $q^j$ and $a^j$ may depend on the history $q^1,a^1,\dots,q^{j-1},a^{j-1}$.)}
\end{algorithmic}
\end{framed}
\vspace{-6mm}
\caption{$\accuracygame_{\sample, \dimension, \queries}[\oracle, \accadv]$}
\end{figure}

\begin{definition}[Accuracy] \label{def:accurateoracle}
An oracle $\oracle$ is \emph{$(\alpha,\beta,\gamma)$-accurate for $\queries$ adaptively chosen queries given $\sample$ samples in $\bits^{\dimension}$} if for every adversary $\accadv$,
$$
\pr{\accuracygame_{\sample, \dimension, \queries}[\oracle, \accadv]}{\textrm{For $(1-\beta)\queries$ choices of $j \in [\queries]$, }\left| \oracle(x, \query^{j}) - \query^j(\dist) \right| \leq \alpha}
\geq 1 - \gamma\,.
$$
As a shorthand, we will say that $\oracle$ is \emph{$(\alpha, \beta)$-accurate for $\queries$ queries} if for every $\sample, \dimension \in \N$, $\oracle$ is $(\alpha,\beta,o_{\sample}(1))$-accurate for $\queries$ queries given $\sample$ samples in $\bits^{\dimension}$.  Here, $\queries$ may depend on $\sample$ and $\dimension$ and $o_{\sample}(1)$ is a function of $\sample$ that tends to $0$.
\end{definition}

We are interested in oracles that are both accurate and computationally efficient.  We say that an oracle $\oracle$ is \emph{computationally efficient} if, when given samples $x_1,\dots,x_{\sample} \in \bits^{\dimension}$ and a query $\query \from \bits^{\dimension} \to [-1,1]$, it runs in time $\poly(\sample, \dimension, |\query|)$.  Here $\query$ will be represented as a circuit that evaluates $\query(x)$ and $|\query|$ denotes the size of this circuit.

\subsection{Encryption Schemes}
Our attack relies on the existence of a semantically secure private-key encryption scheme.  An encryption scheme is a triple of efficient algorithms $(\encgen, \encenc, \encdec)$ with the following syntax:
\begin{itemize}
\item $\encgen$ is a randomized algorithm that takes as input a security parameter $\security$ and outputs a $\security$-bit secret key.  Formally, $\sk \getsr \encgen(1^\security)$.
\item $\encenc$ is a randomized algorithm that takes as input a secret key and a message $m \in \set{-1,0,1}$ and outputs a ciphertext $\ct \in \{0,1\}^{\mathrm{poly}(\lambda)}$.  Formally, $\ct \getsr \encenc(\sk, m)$.
\item $\encdec$ is a deterministic algorithm that takes as input a secret key and a ciphertext $\ct$ and outputs a decrypted message $m'$.  If the ciphertext $\ct$ was an encryption of $m$ under the key $\sk$, then $m' = m$.  Formally, if $\ct \getsr \encenc(\sk, m)$, then $\encdec(\sk, \ct) = m$ with probability $1$.
\end{itemize}

Roughly, security of the encryption scheme asserts that no polynomial time
adversary who does not know the secret key can distinguish encryptions of $m =
0$ from encryptions of $m = 1$, even if the adversary has access to an oracle that returns the
encryption of an arbitrary message under the unknown key.  For
convenience, we will require that this security property holds simultaneously
for an arbitrary polynomial number of secret keys. The existence of an encryption scheme with
this property follows immediately from the existence an ordinary semantically secure encryption scheme.  We 
start with the stronger definition only to simplify our proofs. A secure encryption scheme exists
under the minimal cryptographic assumption that one-way functions exist. The
formal definition of security is not needed until
Section~\ref{sec:securityproofs}.

\subsection{The Attack}

The adversary is specified in Figure~\ref{fig:attack}.  
Observe that $\realgame_{\sample, \dimension}$ is only well defined for pairs $n, d \in N$ for which $1 + \lceil \log_2(2000\sample) \rceil \leq d$, so that there exists a suitable choice of $\security \in \N$.  Through this section we will assume that $\sample = \sample(\dimension)$ is a polynomial in $\dimension$ and that $\dimension$ is a sufficiently large unspecified constant, which ensures that $\realgame_{\sample, \dimension}$ is well defined.

\begin{figure}[ht]
\begin{framed}
\begin{algorithmic}
\STATE{The distribution $\dist$:}
\INDSTATE[1]{Given parameters $\dimension, \sample$, let $\pop = 2000\sample$, let $\security = \dimension - \lceil \log_2(\pop) \rceil$.}
\INDSTATE[1]{Let $(\encgen, \encenc, \encdec)$ be an encryption scheme}
\INDSTATE[1]{For $i \in [\pop]$, let $\sk_{i} \getsr \encgen(1^{\security})$ and let $y_i = (i, \sk_{i}) \in \bits^{\dimension}$.}
\INDSTATE[1]{Let $\dist$ be the uniform distribution over $\set{y_1,\dots,y_{\pop}} \subseteq \bits^{\dimension}$.}
\STATE{}
\STATE{$\oracle$ samples $x_1,\dots,x_{\sample} \getsr \dist$. Let $x = (x_{1},\dots,x_{\sample})$.}
\STATE{Let $S \subseteq [\pop]$ be the set of unique indices $i$ such that $(i, \sk_{i})$ appears in $x$.}
\STATE{}
\STATE{Attack:}
\INDSTATE[1]{Initialise a $\sample$-collusion resilient interactive fingerprinting code $\ifpc$ of length $\length$ for $\pop$ users robust to a $\rob$ fraction of errors with failure probability $\fpcfail=\negl(\sample)$ and false accusation probability $\fpcfalse=1/1000$.}
\INDSTATE[1]{Let $T^1 = \emptyset$.}
\INDSTATE[1]{For $j = 1,\dots,\length = \length(\pop)$:}
\INDSTATE[2]{Let $\FPCcol^{j} \in \pmo^\pop$ be the column given by $\ifpc$.}
\INDSTATE[2]{For $i = 1,\dots,\pop$, let $\ct^{j}_{i} = \encenc(\sk_i, \FPCcol^{j}_{i})$.}
\INDSTATE[2]{Define the query $\query^{j}(i', \sk')$ to be $\encdec(\sk', \ct^j_{i'})$ if $i' \not\in T^{j}$ and $0$ otherwise.}
\INDSTATE[2]{Let $a^{j} = \oracle(x; \query^{j})$ and round $a^j$ to $\pmo$ to obtain $\overline{a}^j$.}
\INDSTATE[2]{Give $\overline{a}^j$ to $\ifpc$ and let $I^{j} \subseteq [\pop]$ be the set of accused users and $T^j = T^{j-1} \cup I^{j}$.}
\end{algorithmic}
\end{framed}
\vspace{-6mm}
\caption{$\realgame_{\sample, \dimension}[\oracle]$}
\label{fig:attack}
\end{figure}

\subsection{Informal Analysis of the Attack} \label{sec:recoveryanalysis}
Before formally analysing the attack, we comment on the overall structure thereof.

At a high level, the attack $\realgame_{\sample,\dimension}[\oracle]$ runs the fingerprinting game $\ifpcgame_{\pop, \collusion, \length}[\fpcadv, \ifpc]$, where the oracle $\oracle$ plays the r\^{o}le of the fingerprinting adversary $\fpcadv$. Each challenge $\FPCcol^\rindex$ issued by $\ifpc$ is passed to the oracle in encrypted form as $q^\rindex$. The oracle must output an approximation $a^j$ to the true answer $$q^\rindex(\mathcal{D})=\frac{1}{\pop} \sum_{i \in [\pop] \setminus T^\rindex} c^\rindex_i.$$  In order to do this, the oracle could decrypt $q^\rindex$ to obtain $\FPCcol^\rindex$ for every $\rindex$. However, the oracle does not have all the necessary secret keys; it only has the secret keys corresponding to its sample $S$. Thus, by the security of the encryption scheme, any efficient oracle effectively can only see $\FPCcol^\rindex_{S\setminus T^\rindex}$. That is to say, if the oracle is computationally efficient, then it has the same restriction as a fingerprinting adversary $\fpcadv$.  Thus, any computationally efficient oracle must lose the fingerprinting game, meaning it cannot answer every query (or even just a $\beta = 1/2 + \Omega(1)$ fraction of the queries) accurately.

One subtly arises since ``accuracy'' for the oracle is defined with respect to the true answer $q^\rindex(\mathcal{D})=\frac{1}{\pop} \sum_{i \in [\pop] \setminus T^\rindex} c^\rindex_i,$ whereas ``accuracy'' in the fingerprinting game is defined with respect to the average over all of $\FPCcol^\rindex$, that is $\frac{1}{\pop} \sum_{i \in [N]} \FPCcol^{\rindex}_{i}$.  We deal with these subtleties by arguing that $T^{\rindex}$, which is the number of users accused by the interactive fingerprinting code prior to the $j$-th query, is small.  Here we use the fact that the fingerprinting code only allows a relatively small number of false accusations $\pop/1000$.  Therefore $|T^{\rindex}| \leq \sample + \pop/1000 \leq \pop/500$.  As a result, the definition of accuracy guaranteed by the oracle will be close enough to the definition of accuracy required for the interactive fingerprinting code to succeed in identifying the sample.

\subsection{Analysis of the Attack} \label{sec:recoveryanalysis}

In this section we prove our main result:

\begin{theorem}[Theorem \ref{thm:main1}] \label{thm:main1-formal}
Assuming one-way functions exist, for all $\rob < 1/2$, there is a function $\length(2000\sample, \rob) = O(\sample^2 / \irob^4)$ such that there is no computationally efficient oracle $\oracle$ that is $(0.99, \rob,1/2)$-accurate for $\length(2000\sample, \rob)$ adaptively chosen queries given $\sample$ samples in $\bits^{\dimension}$.
\end{theorem}

We will start by establishing that the number of falsely accused users is small.  That is, we have $|T^\length \setminus S | \leq \pop/1000$ with high probability.  This condition will follow from the security of the interactive fingerprinting code $\ifpc$.  However, security alone is not enough to guarantee that the number of falsely accused users is small, because security of $\ifpc$ applies to adversaries that only have access to $\FPCcol^{j}_{i}$ for users $i \in S \setminus T^{j}$, whereas the queries to the oracle depend on $\FPCcol^{j}_{i}$ for users $i \not\in S \setminus T^{j}$.  To remedy this problem we rely on the fact entries $\FPCcol^{j}_{i}$ for $i$ outside of $S \setminus T^{j}$ are encrypted under keys $\sk_i$ that are not known to the oracle.  Thus, a computationally efficient oracle ``does not know'' those rows.  We can formalize this argument by comparing $\realgame$ to an $\idealgame$ (Figure~\ref{fig:idealattack}) where these entries are replaced with zeros, and argue that the adversary cannot distinguish between these two attacks without breaking the security of the encryption scheme.
\begin{figure}[ht]
\begin{framed}
\begin{algorithmic}
\STATE{The distribution $\dist$:}
\INDSTATE[1]{Given parameters $\dimension, \sample$, let $\pop = 2000\sample$, and $\security = \dimension - \lceil \log_2(\pop) \rceil$.}
\INDSTATE[1]{Let $(\encgen, \encenc, \encdec)$ be an encryption scheme}
\INDSTATE[1]{For $i \in [\pop]$, let $\sk_{i} \getsr \encgen(1^{\security})$ and let $y_i = (i, \sk_{i}) \in \bits^{\dimension}$.}
\INDSTATE[1]{Let $\dist$ be the uniform distribution over $\set{y_1,\dots,y_{\pop}} \subseteq \bits^{\dimension}$.}
\STATE{}
\STATE{Choose samples $x_1,\dots,x_{\sample} \getsr \dist$, let $x = (x_{1},\dots,x_{\sample})$.}
\STATE{Let $S \subseteq [\pop]$ be the set of unique indices $i$ such that $(i, \sk_{i})$ appears in $x$.}
\STATE{}
\STATE{Recovery phase:}
\INDSTATE[1]{Initialise a $\sample$-collusion resilient interactive fingerprinting code $\ifpc$ of length $\length$ for $\pop$ users robust to a $\rob$ fraction of errors with failure probability $\fpcfail=\negl(\sample)$ and false accusation probability $\fpcfalse=1/1000$.}
\INDSTATE[1]{Let $T^1 = \emptyset$.}
\INDSTATE[1]{For $j = 1,\dots,\length = \length(\pop)$:}
\INDSTATE[2]{Let $\FPCcol^{j} \in \pmo^\pop$ be the column given by $\ifpc$.}
\INDSTATE[2]{\color{DarkRed}{For $i \in S$, let $\ct^{j}_{i} = \encenc(\sk_i, \FPCcol^{j}_{i})$, for $i \in [\pop] \setminus S$, let $\ct^{j}_{i} = \encenc(\sk_i, 0)$.}}
\INDSTATE[2]{Define the query $\query^{j}(i', \sk')$ to be $\encdec(\sk', \ct^j_{i'})$ if $i' \not\in T^{j}$ and $0$ otherwise.}
\INDSTATE[2]{Let $a^{j} = \oracle(x; \query^{j})$ and round $a^j$ to $\pmo$ to obtain $\overline{a}^j$.}
\INDSTATE[2]{Give $\overline{a}^j$ to $\ifpc$ and let $I^{j} \subseteq [\pop]$ be the set of accused users and $T^j = T^{j-1} \cup I^{j}$.}
\end{algorithmic}
\end{framed}
\vspace{-6mm}
\caption{$\idealgame_{\sample, \dimension}[\oracle]$}
\label{fig:idealattack}
\end{figure}

\begin{claim} \label{clm:fewfalseideal}
For every oracle $\oracle$, every polynomial $\sample = \sample(\dimension)$, and every sufficiently large $\dimension \in \N$,
$$
\pr{\idealgame_{\sample, \dimension}[\oracle]}{|T^{\length} \setminus S| > \pop/1000} \leq \negl(\sample)
$$
\end{claim}
\begin{proof}
This follows straightforwardly from a reduction to the security of the fingerprinting code.  Notice that the query $\query^j$ does not depend on any entry $\FPCcol^{j}_{i}$ for $i \not\in S \setminus T^{j-1}$.  Thus, an adversary for the fingerprinting code who has access to $\FPCcol^{j}_{S \setminus T^{j-1}}$ can simulate the view of the oracle.  Since we have for any adversary $\fpcadv$
$$
\pr{\ifpcgame_{\pop, \sample, \length}[\fpcadv, \ifpc]}{\psi^{\length} > (\pop-\sample)\delta} \leq \fpcfail,
$$
we also have
$$
\pr{\idealgame_{\sample, \dimension}[\oracle]}{|T^{\length} \setminus S| > \pop/1000} \leq \negl(\sample),
$$
as desired.
\end{proof}

Now we can argue that an efficient oracle cannot distinguish between the real attack and the ideal attack.  Thus the conclusion that $|T^{\length} \setminus S| \leq \pop/1000$ with high probability must also hold in the real game.
\begin{claim}\label{clm:indist1}
Let $Z_{1}$ be the event
$
\set{|T^{\length} \setminus S| > \pop/1000}.
$
Assume $(\encgen, \encenc, \encdec)$ is a computationally secure encryption scheme and let $\sample = \sample(\dimension)$ be any polynomial.  Then, if $\oracle$ is computationally efficient, for every sufficiently large $\dimension \in \N$
$$
\left| \pr{\idealgame_{\sample, \dimension}[\oracle]}{Z_{1}} - \pr{\realgame_{\sample, \dimension}[\oracle]}{Z_{1}} \right| \leq \negl(\sample)
$$
\end{claim}
The proof is straightforward from the definition of security, and is deferred to Section~\ref{sec:securityproofs}.
Combining Claims~\ref{clm:fewfalseideal} and~\ref{clm:indist1} we easily obtain the following.
\begin{claim} \label{clm:fewfalsereal}
For every computionally efficient oracle $\oracle$, every polynomial $\sample = \sample(\dimension)$, and every sufficiently large $\dimension \in \N$,
$$
\pr{\realgame_{\sample, \dimension}[\oracle]}{|T^{\length} \setminus S| > \pop/1000} \leq \negl(\sample)
$$
\end{claim}

Claim~\ref{clm:fewfalsereal} will be useful because it will allow us to establish that an accurate oracle must give answers that are consistent with the fingerprinting code.  That is, using $\errors^{\length}$ to denote the number of inconsistent answers $\overline{\answer}^1,\dots,\overline{\answer}^{\length}$, we will have $\errors^{\length} \ll \length/2$ with high probability.
\begin{claim}\label{clm:realaccurate}
If $\oracle$ is $(0.99, \rob,1/2)$-accurate for $\length = \length(2000\sample)$ adaptively chosen queries then, for every polynomial $\sample = \sample(\dimension)$ and every sufficiently large $\dimension \in \N$,
$$
\pr{\realgame_{\sample, \dimension}[\oracle]}{\errors^{\length} \leq \rob\length} \geq 1/2 - \negl(\sample)
$$
\end{claim}
\begin{proof}
In the attack, the oracle's input consists of $\sample$ samples from $\dist$, and the total number of queries issued is $\length$.   Therefore, by the assumption that $\oracle$ is $(0.99, \rob,1/2)$-accurate for $\length$ queries, we have
\begin{equation} \label{eq:fd0}
\pr{}{\textrm{For $(1-\rob)\length$ choices of $j \in [\length]$,} \atop \left| \oracle(x, \query^{j}) - \ex{(i, \sk_i) \getsr \dist}{\query^{j}(i, \sk_{i})} \right| \leq 0.99} \geq 1/2.
\end{equation}
Observe that, by construction, for every $j \in [\length]$,
\begin{align}
&\left| \ex{(i, \sk_{i}) \getsr \dist}{\query^{j}(i, \sk_{i})} - \ex{i \in [\pop]}{\FPCcol^{j}_i}\right| \notag \\
={} &\left| \left(\frac{1}{\pop} \sum_{i \in [\pop] \setminus T^{j-1}} \encdec(\sk_{i}, \ct^{j}_{i}) + \frac{1}{\pop} \sum_{i \in T^{j-1}} 0 \right) - \ex{i \in [\pop]}{\FPCcol^{j}_i}\right|  \notag \\
={} &\left| \left(\frac{1}{\pop} \sum_{i \in [\pop] \setminus T^{j-1}} \FPCcol^{j}_{i}\right) - \frac{1}{\pop} \sum_{i \in [\pop]}{\FPCcol^{j}_i}\right| \notag \\
=& \left| -\frac{1}{\pop} \sum_{i \in T^{j-1}} \FPCcol^j_i \right| \notag\\
\leq{} &\frac{\left| T^{j-1} \right|}{\pop} \notag \\
\leq& \frac{|T^{j-1} \setminus S|+|S|}{\pop} \label{eq:fd1}
\end{align}
where the second equality is because by construction $\ct^{j}_{i} \getsr \encenc(\sk_i, \FPCcol^{j}_{i})$ and the inequality is because we have $\FPCcol^{j}_i \in \pmo$.

By Claim~\ref{clm:fewfalsereal}, and the fact that $T^{j-1} \subseteq T^{\length}$, we have
$$
\pr{}{|T^{j-1} \setminus S| > \pop/1000 } \leq \negl(\sample).
$$
Noting that $\pop/1000 +\sample < \pop/500$ and combining with~\eqref{eq:fd1}, we have
\begin{equation}
\pr{}{\forall\; j \in [\length],\; \left| \ex{(i, \sk_{i}) \getsr \dist}{\query^{j}(i, \sk_{i})} - \ex{i \in [\sample]}{\FPCcol^{j}_i}\right| \leq 1/500} \geq 1 - \negl(\sample)
\label{eqn:popaccurate}
\end{equation}

Applying the triangle inequality to~\eqref{eq:fd0} and~\eqref{eqn:popaccurate}, we obtain
\begin{equation} \label{eq:fd2}
\pr{}{\textrm{For $(1-\rob)\length$ choices of $j \in [\length]$,} \atop \left| \oracle(x, \query^{j}) - \ex{i \in
[\pop]}{\FPCcol^{j}_{i}} \right| \leq 0.99+1/500} \geq 1/2 - \negl(\sample).
\end{equation}

Fix a $j \in [\length]$ such that $\answer^{j}$ is $0.99$-accurate for query $\query^j$.  If $\FPCcol^{j}_{i} = 1$ for every $i \in [\pop]$, then $\answer^j = \oracle(x, \query^{j}) \geq 1 - 0.99 - 1/500$, so the rounded answer $\overline{a}^{j} = 1$.  Similarly if $\FPCcol^{j}_{i} = -1$ for every $i \in [\pop]$, $\overline{a}^{j} = -1$.  Therefore there must exist $i \in [\pop]$ so that $\overline{\answer}^{j} = \FPCcol^{j}_{i}$.  Thus there are $(1-\rob)\length$ choices of $j \in [\length]$ for which this condition holds, so the number of errors $\errors^{\length}$ is at most $\rob\length$.  This completes the proof of the claim.
\end{proof}

As before, we can argue that the real attack and the ideal attack are computationally indistinguishable, and thus the oracle must also give consistent answers in the ideal attack.
\begin{claim}\label{clm:indist2}
Let $Z_{2}$ be the event
$
\set{\errors^{\length} \leq \rob\length}.
$
Assume $(\encgen, \encenc, \encdec)$ is a computationally secure encryption scheme and let $\sample = \sample(\dimension)$ be any polynomial.  Then if $\oracle$ is computationally efficient, for every $\dimension \in \N$
$$
\left| \pr{\idealgame_{\sample, \dimension}[\oracle]}{Z_{2}} - \pr{\realgame_{\sample, \dimension}[\oracle]}{Z_{2}} \right| \leq \negl(\sample)
$$
\end{claim}
The proof is straightforward from the definition of security, and is deferred to Section~\ref{sec:securityproofs}.
Combining Claims~\ref{clm:realaccurate} and~\ref{clm:indist2} we easily obtain the following.
\begin{claim}\label{clm:idealaccurate}
If $\oracle$ computationally efficient and $(0.99, \rob,1/2)$-accurate for $\length = \length(2000\sample)$ adaptively chosen queries then for every polynomial $\sample = \sample(\dimension)$ and every sufficiently large $\dimension \in \N$,
$$
\pr{\idealgame_{\sample, \dimension}[\oracle]}{\errors^{\length} \leq \rob\length} \geq 1/2 - \negl(\sample).
$$
\end{claim}

However, the conclusion of~\ref{clm:idealaccurate} can easily be seen to lead to a contradiction, because the security of the fingerprinting code assures that no attacker who only has access to $\FPCcol^{j}_{S \setminus T^{j-1}}$ in each round $j = 1,\dots,\length$ can give answers that are consistent for $(1-\rob)\length$ of the columns $\FPCcol^{j}$.  Thus, we have
\begin{claim} \label{clm:idealnotaccurate}
For every oracle $\oracle$, every polynomial $\sample = \sample(\dimension)$, and every sufficiently large $\dimension \in \N$,
$$
\pr{\idealgame_{\sample, \dimension}[\oracle]}{\errors^{\length} \leq \rob\length} \leq \negl(\sample)
$$
\end{claim} 

Putting the above claims together, we obtain the main theorem:

\begin{proof}[Proof of Theorem \ref{thm:main1-formal}]
Assume for the sake of contradiction that there were such an oracle. 
Theorem \ref{thm:ifpcthm} implies that an interactive fingerprinting code of length $O(\sample^2 / \irob^4)$ exists, so the attack can be carried out.
By Claim~\ref{clm:idealaccurate} we would have
$$
\pr{\idealgame_{\sample, \dimension}[\oracle]}{\errors^{\length} \leq \rob\length} \geq 1/2 - \negl(\sample).
$$
But, by Claim~\ref{clm:idealnotaccurate} we have
$$
\pr{\idealgame_{\sample, \dimension}[\oracle]}{\errors^{\length} \leq \rob\length} \leq \negl(\sample),
$$
which is a contradiction.
\end{proof}

Note that the constants in the $(0.99,\beta,1/2)$-accuracy assumption are arbitrary and have only been fixed for simplicity.

\subsection{An Information-Theoretic Lower Bound} \label{sec:infotheoretic}
As in~\cite{HardtU14}, we observe that the techniques underlying our computational hardness result can also be used to prove an information-theoretic lower bound when the dimension of the data is large.  At a high level, the argument uses the fact that the encryption scheme we rely on only needs to satisfy relatively weak security properties, specifically security for at most $O(n^2)$ messages.  This security property can actually be achieved against computationally unbounded adversaries provided that the length of the secret keys is $O(n^2)$.  As a result, our lower bound can be made to hold against computationally unbounded oracles, but since the secret keys have length $O(n^2)$, we will require $d = O(n^2)$.  We refer the reader to~\cite{HardtU14} for a slightly more detailed discussion, and simply state the following result.
\begin{theorem}[Theorem \ref{thm:main2}]
For all $\rob < 1/2$, there is a function $\length(2000\sample, \rob) = O(\sample^2 / \irob^4)$ such that there is no oracle $\oracle$ (even one that is computationally unbounded) that is $(0.99, \rob,1/2)$-accurate for $\length(2000\sample, \rob)$ adaptively chosen queries given $\sample$ samples in $\bits^{\dimension}$ when $\dimension \geq \length(2000\sample, \rob)$.
\end{theorem}

\section{Hardness of Avoiding Blatant Non Privacy} \label{sec:nonprivacy}
In this section we show how our arguments also imply that computationally efficient oracles that guarantee accuracy for adaptively chosen statistical queries must be blatantly non-private.

\subsection{Blatant Non Privacy and Sample Accuracy}
Before we can define blatant non-privacy, we need to define a notion of accuracy that is more appropriate for the application to privacy.  In contrast to Definition~\ref{def:accurateoracle} where accuracy is defined with respect to the distribution, here we define accurate with respect to the sample itself.  With this change in mind, we model blatant non-privacy via the following game.

\begin{figure}[ht]
\begin{framed}
\begin{algorithmic}
\STATE{$\bnpadv$ chooses a set $y = \{y_1,\dots,y_{2\sample}\} \subseteq \bits^{\dimension}$}
\STATE{Sample a random subsample $x \subseteq_{\mbox \tiny R} y$ of size $\sample$}
\STATE{For $j = 1,\dots,\queries$}
\INDSTATE[1]{$\bnpadv$ outputs a query $\query^j$}
\INDSTATE[1]{$\oracle(x, \query^{j})$ outputs $a^{j}$}
\INDSTATE[1]{(As $\bnpadv$ and $\oracle$ are stateful, $\query^j$ and $a^j$ may depend on $\query^1, a^1, \dots, \query^{j-1},a^{j-1}$.)}
\STATE{$\bnpadv$ outputs a set $x' \subseteq y$}
\end{algorithmic}
\end{framed}
\vspace{-6mm}
\caption{$\bnpgame_{\sample, \dimension}[\oracle, \bnpadv]$}
\end{figure}

\begin{definition}
An oracle $\oracle$ is \emph{$(\alpha,\beta, \gamma)$-sample-accurate for $\queries$ adaptively chosen queries given $\sample$ samples in $\bits^{\dimension}$} if for every adversary $\bnpadv$,
$$
\pr{\bnpgame_{\sample, \dimension, \queries}[\oracle, \bnpadv]}{\textrm{For $(1-\beta)\queries$ choices of $j \in [\queries]$, $\left| \oracle(x, \query^{j}) - \query^j(x) \right| \leq \alpha$}}
\geq 1 - \gamma\,
$$
where $\query(x) = \frac{1}{\sample} \sum_{i \in [\sample]} \query(x_i)$ is the average over the \emph{sample}.

As a shorthand, we will say that $\oracle$ is \emph{$(\alpha, \beta)$-sample-accurate for $\queries$ queries} if for every $\sample, \dimension \in \N$, $\oracle$ is $(\alpha,\beta,o_{\sample}(1))$-accurate for $\queries$ queries given $\sample$ samples in $\bits^{\dimension}$.  Here, $\queries$ may depend on $\sample$ and $\dimension$ and $o_{\sample}(1)$ is a function of $\sample$ that tends to $0$.
\end{definition}

\begin{definition}
An oracle $\oracle$ is \emph{blatantly non-private} if there exists an adversary $\bnpadv$ such that
$$
\pr{\bnpgame_{\sample, \dimension, \queries}[\oracle, \bnpadv]}{|x \triangle x'| >  \sample / 100} \leq o_{\sample}(1)
$$

\end{definition}

\subsection{Lower Bounds}
In this section we show the following theorem
\begin{theorem} \label{thm:nonprivacy1}
Assuming one-way functions exist, any computationally efficient oracle $\oracle$ that gives accurate answers to $O(\sample^2)$ adaptively chosen queries is blatantly non-private.
\end{theorem}

The attack is defined in Figure \ref{fig:bnpattack}. Therein $\ifpc$ is a $\sample$-collusion-resilient interactive fingerprinting code of length $\length$ for $\pop=2\sample$ users robust to a $\rob$ fraction of errors with false accusation probability $\fpcfalse=1/20000$. And $(\encgen, \encenc, \encdec)$ is a computationally secure encryption scheme.

\begin{figure}[ht]
\begin{framed}
\begin{algorithmic}
\STATE{The set $y$:}
\INDSTATE[1]{Given parameters $\dimension, \sample$, let $\security = \dimension - \lceil \log_2(2\sample) \rceil$.}
\INDSTATE[1]{For $i \in [2\sample]$, let $\sk_{i} \getsr \encgen(1^{\security})$ and let $y_i = (i, \sk_{i})$.}
\STATE{}
\STATE{Attack:}
\INDSTATE[1]{Let $T^1 = \emptyset$.}
\INDSTATE[1]{For $j = 1,\dots,\length = \length(2\sample)$:}
\INDSTATE[2]{Let $\FPCcol^{j} \in \pmo^{2\sample}$ be the column given by $\ifpc$.}
\INDSTATE[2]{For $i = 1,\dots,2\sample$, let $\ct^{j}_{i} = \encenc(\sk_i, \FPCcol^{j}_{i})$.}
\INDSTATE[2]{Define the query $\query^{j}(i', \sk')$ to be $\encdec(\sk', \ct^j_{i'})$ if $i' \not\in T^{j}$ and $0$ otherwise.}
\INDSTATE[2]{Let $a^{j} = \oracle(x; \query^{j})$ and round $(\sample/(\sample - |T^{j-1}|))\answer^{j}$ to $\pmo$ to obtain $\overline{a}^j$.}
\INDSTATE[2]{Give $\overline{a}^j$ to $\ifpc$ and let $I^{j} \subseteq [\pop]$ be the set of accused users and $T^j = T^{j-1} \cup I^{j}$.}
\INDSTATE[2]{If $|T^{j}| > 499\sample/500$, let $L = j$, halt, and output $x' = \{ y_i : i \in T^{L} \}$.}
\INDSTATE[1]{Let $L = \length$, and output $x' = \{ y_i : i \in T^{L} \}$.}
\end{algorithmic}
\end{framed}
\vspace{-6mm}
\caption{$\bnpattack_{\sample, \dimension}[\oracle]$}
\label{fig:bnpattack}
\end{figure}

We will start by establishing that the number of falsely accused users is small.  That is, we have $|T^{L}\setminus x| \leq \sample/10000$ with high probability.  As in Section~\ref{sec:falsedisc}, this condition will follow from the security of the interactive fingerprinting code $\ifpc$ combined with the security of the encryption scheme, via the introduction of an ``ideal attack'' (Figure~\ref{fig:idealbnpattack}).

\begin{figure}[ht]
\begin{framed}
\begin{algorithmic}
\STATE{The set $y$:}
\INDSTATE[1]{Given parameters $\dimension, \sample$, let $\security = \dimension - \lceil \log_2(2\sample) \rceil$.}
\INDSTATE[1]{For $i \in [2\sample]$, let $\sk_{i} \getsr \encgen(1^{\security})$ and let $y_i = (i, \sk_{i})$.}
\STATE{}
\STATE{Attack:}
\INDSTATE[1]{Let $T^1 = \emptyset$.}
\INDSTATE[1]{For $j = 1,\dots,\length = \length(2\sample)$:}
\INDSTATE[2]{Let $\FPCcol^{j} \in \pmo^{2\sample}$ be the column given by $\ifpc$.}
\INDSTATE[2]{For $i = 1,\dots,2\sample$, let $\ct^{j}_{i} = \encenc(\sk_i, \FPCcol^{j}_{i})$.}
\INDSTATE[2]{\color{DarkRed}{For $i \in S$, let $\ct^{j}_{i} = \encenc(\sk_i, \FPCcol^{j}_{i})$, for $i \in [\pop] \setminus x$, let $\ct^{j}_{i} = \encenc(\sk_i, 0)$.}}
\INDSTATE[2]{Let $a^{j} = \oracle(x; \query^{j})$ and round $(\sample/(\sample - |T^{j-1}|))\answer^{j}$ to $\pmo$ to obtain $\overline{a}^j$.}
\INDSTATE[2]{Give $\overline{a}^j$ to $\ifpc$ and let $I^{j} \subseteq [\pop]$ be the set of accused users and $T^j = T^{j-1} \cup I^{j}$.}
\INDSTATE[2]{If $|T^{j}| > 499\sample/500$, let $L = j$, halt, and output $x' = \{ y_i : i \in T^{L} \}$.}
\INDSTATE[1]{Let $L = \length$, and output $x' = \{ y_i : i \in T^{L} \}$.}
\end{algorithmic}
\end{framed}
\vspace{-6mm}
\caption{$\idealbnpattack_{\sample, \dimension}[\oracle]$}
\label{fig:idealbnpattack}
\end{figure}

\begin{claim} \label{clm:bnpfewfalseideal}
For every oracle $\oracle$, every polynomial $\sample = \sample(\dimension)$, and every sufficiently large $\dimension \in \N$,
$$
\pr{\idealgame_{\sample, \dimension}[\oracle]}{|T^{L} \setminus x| > \sample/10000} \leq \negl(\sample)
$$
\end{claim}
\begin{proof}
This follows straightforwardly from a reduction to the security of the fingerprinting code.  Notice that since the query $\query^j$ does not depend on any entry $\FPCcol^{j}_{i}$ for $i \not\in x \setminus T^{j-1}$.  Thus, an adversary for the fingerprinting code who has access to $\FPCcol^{j}_{x \setminus T^{j-1}}$ can simulate the view of the oracle.  Since we have for any adversary $\fpcadv$
$$
\pr{\ifpcgame_{\pop, \sample, \length}[\fpcadv, \ifpc]}{\psi^{\length} > \pop/20000} \leq \negl(\sample),
$$
we also have
$$
\pr{\idealbnpattack_{\sample, \dimension}[\oracle]}{|T^{L} \setminus x | > \sample/10000} \leq \negl(\sample),
$$
where we have used the fact that $|T^L \setminus x | = \falseaccs^{L} \leq \falseaccs^{\length}$.  This completes the proof.
\end{proof}

Now we can argue that an efficient oracle cannot distinguish between the real attack and the ideal attack.  Thus the conclusion that $|T^{L} \setminus x| \leq \sample/10000$ with high probability must also hold in the real game.
\begin{claim}\label{clm:bnpindist1}
Let $Z_{1}$ be the event
$
\set{|T^{L} \setminus x| > \sample/10000}
$.
Assume $(\encgen, \encenc, \encdec)$ is a computationally secure encryption scheme and let $\sample = \sample(\dimension)$ be any polynomial.  Then if $\oracle$ is computationally efficient, for every $\dimension \in \N$
$$
\left| \pr{\idealbnpattack_{\sample, \dimension}[\oracle]}{Z_{1}} - \pr{\bnpattack_{\sample, \dimension}[\oracle]}{Z_{1}} \right| \leq \negl(\sample)
$$
\end{claim}
The proof is straightforward from the definition of security, and is deferred to Section~\ref{sec:securityproofs}.
Combining Claims~\ref{clm:bnpfewfalseideal} and~\ref{clm:bnpindist1} we easily obtain the following.
\begin{claim} \label{clm:bnpfewfalsereal}
For every computionally efficient oracle $\oracle$, every polynomial $\sample = \sample(\dimension)$, and every sufficiently large $\dimension \in \N$,
$$
\pr{\bnpattack_{\sample, \dimension}[\oracle]}{|T^{L} \setminus x| > \sample/10000} \leq \negl(\sample)
$$
\end{claim}

By Claim~\ref{clm:bnpfewfalsereal} we have $|x' \setminus x| \leq \sample/10000$.  Now, in order to show $|x' \triangle x| \leq \sample/100$, it suffices to show that $|x \setminus x'| \leq \sample/200$.  In order to do so we begin with the following claim, which establishes that if the oracle $\oracle$ is sufficiently accurate, and $|x \setminus T^{j-1}| \leq \sample/200$, then the oracle returns a consistent answer to the query $\query^{j}$.  Recalling that we use $\theta^{j}$ to denote the number of rounded answers $\overline{\answer}^{k}$ for $1 \leq k \leq j$ that are inconsistent with $\FPCcol^{j}$, we can state the following claim.
\begin{claim}\label{clm:bnprealaccurate}
If $\oracle$ is $(1/1000, \beta, 1/2)$-sample-accurate for $\length = \length(2\sample,\rob)$ adaptively chosen queries then for every polynomial $\sample = \sample(\dimension)$, every sufficiently large $\dimension \in \N$, 
$$
\pr{\bnpattack_{\sample, \dimension}[\oracle]}{\theta^L \leq \beta L } \geq 1/2.
$$
\end{claim}
\begin{proof}
Observe that, by construction, for every $j \in [\length]$,
\begin{align}
\ex{x_i \in x}{\query^j(x_i)}
&={}\frac{1}{\sample} \left( \sum_{i \in (x \setminus T^{j-1})} \FPCcol^{j}_i + \sum_{i \in (x \cap T^{j-1})} 0 \right) \notag \\
&={} \ex{i \in (x \setminus T^{j-1})}{\FPCcol^{j}_{i}} \cdot \left(\frac{|x \setminus T^{j-1}|}{\sample}\right) \notag \notag
\end{align}
After renormalizing by $(\sample / \sample - |T^{j-1}|)$ we have
\begin{align}
&\left(\frac{\sample}{\sample - |T^{j-1}|} \right) \cdot \ex{i \in x}{\query^j(x_i)} \notag \\
={} &\ex{i \in (x \setminus T^{j-1})}{\FPCcol^{j}_{i}} \cdot \left(\frac{\sample}{\sample - |T^{j-1}|} \right) \cdot \left(\frac{|x \setminus T^{j-1}|}{\sample}\right) \notag \\
={} &\ex{i \in (x \setminus T^{j-1})}{\FPCcol^{j}_{i}} \cdot \left(\frac{\sample - |T^{j-1}| + |T^{j-1} \setminus x|}{\sample - |T^{j-1}|} \right) \notag \\
={} &\ex{i \in (x \setminus T^{j-1})}{\FPCcol^{j}_{i}} \cdot \left(1 + \frac{|T^{j-1} \setminus x|}{\sample - |T^{j-1}|} \right) \notag \notag
\end{align}
Since $0 \leq |T^{j-1} \setminus x| \leq \sample/10000$ (by Claim~\ref{clm:bnpfewfalsereal}), and since the algorithm terminates unless $|T^{j-1}| \leq 499\sample/500$, we obtain
\begin{align}
&\ex{i \in (x \setminus T^{j-1})}{\FPCcol^{j}_{i}} \leq \left(\frac{\sample}{\sample - |T^{j-1}|} \right)\cdot \ex{i \in x}{\query^j(x_i)} \leq \frac{21}{20} \cdot \ex{i \in (x \setminus T^{j-1})}{\FPCcol^{j}_{i}} \notag \\
\Longrightarrow{} &\left| \left(\frac{\sample}{\sample - |T^{j-1}|} \right)\cdot \ex{i \in x}{\query^j(x_i)} - \ex{i \in (x \setminus T^{j-1})}{\FPCcol^{j}_{i}}\right| \leq \frac{1}{20} \label{eq:bnp0}
\end{align}
By the assumption that $\oracle$ is $(1/1000, \rob, 1/2)$-sample-accurate, we have that, with probability at least $1/2$, for $(1-\beta) L$ choices of $j \in [L]$,
\begin{equation} \label{eq:bnp1}
\left| \answer^{j} - \ex{i \in x}{\query^j(x_i)} \right| \leq 1/1000.
\end{equation}
Now, combining~\eqref{eq:bnp0} and~\eqref{eq:bnp1}, we have
\begin{align}
&\left| \left(\frac{\sample}{\sample - |T^{j-1}|} \right) \cdot \answer^j - \ex{i \in (x \setminus T^{j-1})}{\FPCcol^{j}_{i}} \right| \notag \\
\leq{} &\left|  \left(\frac{\sample}{\sample - |T^{j-1}|} \right) \cdot \ex{i \in x}{\query^j(x_i)} - \ex{i \in (x \setminus T^{j-1})}{\FPCcol^{j}_{i}} \right| + \left|\frac{\sample}{\sample - |T^{j-1}|}  \cdot \frac{1}{1000} \right| \leq{} \frac{1}{20} + \frac{1}{2} \leq \frac{2}{3} \label{eq:bnp2}
\end{align}
for $(1-\beta)L$ choices of $j \in [L]$.

Finally, observe that if $\FPCcol^{j}_{i} = 1$ for every $i \in [2\sample]$, then we have $$\ex{i \in (x \setminus T^{j-1})}{\FPCcol^{j}_{i}} = 1,$$ and by~\eqref{eq:bnp2} we have $(\sample / (\sample - |T^{j}|)) \answer^{j} \geq 1 - 2/3 = 1/3$.  Thus, the rounded answer $\overline{\answer}^{j} = 1$.  Similarly, if $\FPCcol^{j}_{i} = -1$ for every $i \in [2\sample]$, then we have $\overline{\answer}^{j} = -1$.  This completes the proof of the claim.
\end{proof}

As before, we can argue that the real attack and the ideal attack are computationally indistinguishable, and thus the oracle must also give consistent answers in the ideal attack.
\begin{claim}\label{clm:bnpindist2}
Let $Z_{2}$ be the event
$
\set{\errors^{L} \leq \beta L}.
$
Assume $(\encgen, \encenc, \encdec)$ is a computationally secure encryption scheme and let $\sample = \sample(\dimension)$ be any polynomial.  Then if $\oracle$ is computationally efficient, for every $\dimension \in \N$
$$
\left| \pr{\idealbnpattack_{\sample, \dimension}[\oracle]}{Z_{2}} - \pr{\bnpattack_{\sample, \dimension}[\oracle]}{Z_{2}} \right| \leq \negl(\sample).
$$
\end{claim}
The proof is straightforward from the definition of security, and is deferred to Section~\ref{sec:securityproofs}.
Combining Claims~\ref{clm:bnprealaccurate} and~\ref{clm:bnpindist2} we easily obtain the following.
\begin{claim}\label{clm:bnpidealaccurate}
If $\oracle$ iscomputationally efficient and $(1/1000,\beta, 1/2)$-accurate for $\length = \length(2\sample, \rob)$ adaptively chosen queries then for every polynomial $\sample = \sample(\dimension)$ and every sufficiently large $\dimension \in \N$,
$$
\pr{\idealbnpattack_{\sample, \dimension}[\oracle]}{\errors^{L} \leq \beta L} \geq 1/2 - \negl(\sample).
$$
\end{claim}
We can use Claim~\ref{clm:bnpidealaccurate} to derive a contradiction.  To do so we use the fact that the security of the fingerprinting code assures that no attacker who only has access to $\FPCcol^{j}_{x \setminus T^{j-1}}$ in each round $j = 1,\dots,\length$ can give answers that are consistent for all $\length$ of the columns $\FPCcol^{j}$.  Thus, we have
\begin{claim} \label{clm:bnpidealnotaccurate}
For every oracle $\oracle$, every polynomial $\sample = \sample(\dimension)$, and every sufficiently large $\dimension \in \N$, if $L = \length$
$$
\pr{\idealbnpattack_{\sample, \dimension}[\oracle]}{\errors^{\length} \leq \beta \length} \leq \negl(\sample)
$$
\end{claim} 

Putting it together, we obtain the following theorem.
\begin{theorem}
Assuming one-way functions exist, for every $\rob < 1/2$, there is a function $\length(2\sample,\rob) = O(\sample^2/\irob^4)$ such that there is no computationally efficient oracle $\oracle$ that is $(1/1000, \rob,1/2)$-accurate for $\length(2\sample)$ adaptively chosen queries given $\sample$ samples in $\bits^{\dimension}$.
\end{theorem}
\begin{proof}
Assume for the sake of contradiction that there were such an oracle.  Now consider two cases.  First consider the case that $L < \length$, which means the algorithm has terminated early due to the condition $|T^{L}| \geq 499\sample/500$ being reached.  In this case we have $|x'| = |T^{L}| \geq 499\sample/500$.  However, by Claim~\ref{clm:bnpfewfalseideal}, we have that $|x' \setminus x| \leq \sample/10000$.  Therefore we have $$|x \triangle x'| = |x|-|x'|+2|x'\setminus x|  \leq \frac{\sample}{500} + \frac{2\sample}{10000} \leq \frac{\sample}{200},$$ as desired.

Now consider the case where $L = \ell$, meaning the algorithm does not terminate early.  In this case, by Claim~\ref{clm:bnpidealaccurate} we have
$$
\pr{\idealbnpattack_{\sample, \dimension}[\oracle]}{\errors^{\length}  \leq \rob L} \geq 1/2 - \negl(\sample),
$$
but by Claim~\ref{clm:bnpidealnotaccurate} we have
$$
\pr{\idealbnpattack_{\sample, \dimension}[\oracle]}{\errors^{\length} \leq \rob L} \leq \negl(\sample),
$$
which is a contradiction.  This completes the proof of the theorem.
\end{proof}

\subsection{An Information-Theoretic Lower Bound}
As we did in Section~\ref{sec:infotheoretic}, we can prove an information-theoretic analogue of our hardness result for avoiding blatant non-privacy.
\begin{theorem}
There is a function $\length(2\sample,\rob) = O(\sample^2/\irob^4)$ such that there is no oracle $\oracle$ (even a computationally unbounded one) that is $(1/1000, \rob,1/2)$-accurate for $\length(2\sample,\rob)$ adaptively chosen queries given $\sample$ samples in $\bits^{\dimension}$ where $\dimension \geq \length(2\sample,\rob)$.
\end{theorem}
The proof is essentially identical to what is sketched in Section~\ref{sec:infotheoretic}.

\section*{Acknowledgements} \addcontentsline{toc}{section}{Acknowledgements}
We thank Moritz Hardt and Salil Vadhan for insightful discussions during the early stages of this work. We also thank Thijs Laarhoven for bringing his work on interactive fingerprinting codes to our attention.

\addcontentsline{toc}{section}{References}
\bibliographystyle{alpha}
\bibliography{references}
\appendix

\section{Security Reductions from Sections~\ref{sec:falsedisc} and~\ref{sec:nonprivacy}} \label{sec:securityproofs}

In Section~\ref{sec:falsedisc} we made several claims comparing the probability of events in $\realgame$ to the probability of events in $\idealgame$.  Each of these claims follow from the assumed security of the encryption scheme.  In this section we restate and prove these claims.  Since the claims are all of a similar nature, the proof will be somewhat modular.  The claims in Section~\ref{sec:nonprivacy} relating $\bnpattack$ to $\idealbnpattack$ can be proven in an essentially identical fashion, and we omit these proofs for brevity.

Before we begin recall the formal definition of security of an encryption scheme.  Security is defined via a pair of oracles $\encoracle_0$ and $\encoracle_1$.  $\encoracle_1(\sk_{1},\dots,\sk_{\pop}, \cdot)$ takes as input the index of a key $i \in [\pop]$ and a message $m$ and returns $\encenc(\sk_{i}, m)$, whereas $\encoracle_{0}(\sk_{1},\dots,\sk_{\pop}, \cdot)$ takes the same input but returns $\encenc(\sk_{i}, 0)$.  The security of the encryption scheme asserts that for randomly chosen secret keys, no computationally efficient adversary can tell whether or not it is interacting with $\encoracle_{0}$ or $\encoracle_{1}$.

\begin{definition}
An encryption scheme $(\encgen, \encenc, \encdec)$ is $\emph{secure}$ if for every polynomial $\pop = \pop(\security)$, and every $\poly(\security)$-time adversary $\encadv$, if $\sk_{1},\dots,\sk_{\pop} \getsr \encgen(1^{\security})$
\begin{equation*}
\left| \pr{}{\encadv^{\encoracle_{0}(\sk_{1},\dots,\sk_{\pop}, \cdot)} = 1} 
- \pr{}{\encadv^{\encoracle_{1}(\sk_{1},\dots,\sk_{\pop}, \cdot)} = 1} \right| = \negl(\security)
\end{equation*}
\end{definition}

We now restate the relevant claims from Section~\ref{sec:falsedisc}.

\begin{claim}[Claim~\ref{clm:indist1} Restated] \label{clm:indist1restated}
Let $Z_{1}$ be the event
$
\set{\psi^{\length} > \pop/8}.
$
Assume $(\encgen, \encenc, \encdec)$ is a computationally secure encryption scheme and let $\sample = \sample(\dimension)$ be any polynomial.  Then if $\oracle$ is computationally efficient, for every $\dimension \in \N$
$$
\left| \pr{\idealgame_{\sample, \dimension}[\oracle]}{Z_{1}} - \pr{\realgame_{\sample, \dimension}[\oracle]}{Z_{1}} \right| \leq \negl(\sample)
$$
\end{claim}

\begin{claim}[Claim~\ref{clm:indist2} Restated] \label{clm:indist2restated}
Let $Z_{2}$ be the event
$
\set{\errors^{\length} \leq \rob\length}.
$
Assume $(\encgen, \encenc, \encdec)$ is a computationally secure encryption scheme and let $\sample = \sample(\dimension)$ be any polynomial.  Then if $\oracle$ is computationally efficient, for every $\dimension \in \N$
$$
\left| \pr{\idealgame_{\sample, \dimension}[\oracle]}{Z_{2}} - \pr{\realgame_{\sample, \dimension}[\oracle]}{Z_{2}} \right| \leq \negl(\sample)
$$
\end{claim}

To prove both of these claims, for $c \in \{1,2\}$ we construct an adversary $\encadv_{c}$ that will attempt to use $\oracle$ to break the security of the encryption.  We construct $\encadv_{c}$ in such a way that its advantage in breaking the security of encryption is precisely the difference in the probability of the event $Z_{c}$ between $\realgame$ and $\idealgame$, which implies that the difference in probabilities is negligible.  The simulator is given in Figure~\ref{fig:simulator}

\begin{figure}[ht]
\begin{framed}
\begin{algorithmic}
\STATE{Simulate constructing and sampling from $\dist$:}
\INDSTATE[1]{Given parameters $\dimension, \sample$, let $\pop = 2000\sample$, let $\security = \dimension - \lceil \log_2(2000\sample) \rceil$.}
\INDSTATE[1]{Sample users $u_{1},\dots,u_{\sample} \getsr [\pop]$, let $S$ be the set of distinct users in the sample.}
\INDSTATE[1]{Choose new keys $\sk_{i} \getsr \encgen(1^{\security})$ for $i \in S$.}
\INDSTATE[1]{For $i \in S$, let $x_{i} = (u_{i}, \sk_{u_{i}})$, let $x = (x_1,\dots,x_{\sample})$.}
\STATE{}
\STATE{Simulate the attack:}
\INDSTATE[1]{Let $T^1 = \emptyset$.}
\INDSTATE[1]{For $j = 1,\dots,\length = \length(\pop)$:}
\INDSTATE[2]{Let $\FPCcol^{j}$ be the column given by $\ifpc$.}
\INDSTATE[2]{For $i = 1,\dots,\pop$:}
\INDSTATE[3]{If $i \in S$, let $\ct^{j}_{i} = \encenc(\sk_i, \FPCcol^{j}_{i})$, otherwise as $\encoracle$ for an encryption of $\FPCcol^{j}_{i}$ under}
\INDSTATE[3]{key $\ssk_{i}$, that is $\ct^{j}_{i} = \encoracle_b(\ssk_{1},\dots,\ssk_{\pop}, i, \FPCcol^{j}_{i})$.}
\INDSTATE[2]{Define the query $\query^{j}(i', \sk')$ to be $\encdec(\sk', \ct^j_{i'})$ if $i' \not\in T^{j}$ and $0$ otherwise.}
\INDSTATE[2]{Let $a^{j} = \oracle(x; \query^{j})$ and round $a^j$ to $\pmo$ to obtain $\overline{a}^j$.}
\INDSTATE[2]{Give $\overline{a}^j$ to $\ifpc$ and let $I^{j} \subseteq [\pop]$ be the set of accused users and $T^j = T^{j-1} \cup I^{j}$.}
\STATE{}
\STATE{Output $1$ if and only if the event $Z_{c}$ occurs}
\end{algorithmic}
\end{framed}
\vspace{-6mm}
\caption{$\encadv^{\encoracle_{b}(\ssk_{1},\dots,\ssk_{\pop}, \cdot)}_{c, \sample, \dimension}$}
\label{fig:simulator}
\end{figure}

\begin{proof}[Proof of Claims \ref{clm:indist1restated}, \ref{clm:indist2restated}]
First, observe that for $c \in \set{1,2}$, $\encadv_{c}$ is computationally efficient as long as $\ifpc$ and $\oracle$ are both computationally efficient.  It is not hard to see that our construction $\ifpc$ is efficient and efficiency of $\oracle$ is an assumption of the claim.  Also notice $\encadv$ can determine whether $Z_{c}$ has occurred efficiently.

Now we observe that when the oracle is $\encoracle_1$ (the oracle that takes as input $i$ and $m$ and returns $\encenc(\ssk_{i}, m)$), and $\ssk_{1},\dots,\ssk_{\pop}$ are chosen randomly from $\encgen(1^{\security})$, then the view of the oracle is identical to $\realgame_{\sample, \dimension}[\oracle]$.  Specifically, the oracle holds a random sample of pairs $(i, \sk_{i})$ and is shown queries that are encryptions either under keys it knows or random unknown keys.  Moreover, the messages being encrypted are chosen from the same distribution.  On the other hand, when the oracle is $\encoracle_0$ (the oracle that takes as input $i$ and $\ct$ and returns $\encenc(\ssk_{i}, 0)$), then the view of the oracle is identical to $\realgame_{\sample, \dimension}[\oracle]$.  Thus we have that for $c \in \{1,2\}$,
\begin{align*}
&\left| \pr{\idealgame_{\sample, \dimension}[\oracle]}{Z_{c}} - \pr{\realgame_{\sample, \dimension}[\oracle]}{Z_{c}} \right| \\
={} &\left| \pr{\ssk_{1},\dots,\ssk_{\pop} \getsr \encgen(1^{\security})}{\encadv_{c, \sample, \dimension}^{\encoracle_0(\ssk_{1},\dots,\ssk_{\pop},\cdot)} = 1} -  \pr{\ssk_{1},\dots,\ssk_{\pop} \getsr \encgen(1^{\security})}{\encadv_{c, \sample, \dimension}^{\encoracle_1(\ssk_{1},\dots,\ssk_{\pop},\cdot)} = 1} \right|
={} \negl(\security) = \negl(\dimension)
\end{align*}
The last equality holds because we have chosen $\pop = 2000\sample(\dimension) = \poly(\dimension)$, and therefore we have $\security = \dimension - \lceil \log \pop \rceil = \dimension - O(\log \dimension)$.  This completes the proof of both claims.
\end{proof}

\end{document}